\documentclass{IEEEtran}
\usepackage{amsmath, amssymb, amsthm}

\usepackage{cleveref}
\usepackage{siunitx}
\usepackage{graphicx, subfigure}
\usepackage{booktabs,multirow}
\usepackage{comment}
\usepackage{tikz}
\usepackage{pgfplots}
\pgfplotsset{compat=1.18} 

\usepackage{pgfplotstable}
\usepgfplotslibrary{groupplots}  
\usepgfplotslibrary{colorbrewer} 
\usepgfplotslibrary{colormaps}   
\usetikzlibrary{plotmarks}       
\usetikzlibrary{positioning}     

\usetikzlibrary{fit,backgrounds,positioning}
\usetikzlibrary{arrows.meta}
\tikzset{>=latex}
\graphicspath{{./Figs/}}

\newcommand{\G}{\mathsf{G}}

\renewcommand{\Re}{\ensuremath{\mathbb{R}}}

\newcommand{\U}{\ensuremath{\mathsf{U}}}
\newcommand{\SO}[1]{\ensuremath{\mathsf{SO(#1)}}}
\newcommand{\T}{\ensuremath{\mathsf{T}}}
\renewcommand{\L}{\ensuremath{\mathsf{L}}}
\newcommand{\R}{\ensuremath{\mathsf{R}}}

\newcommand{\Ad}{\ensuremath{\mathrm{Ad}}}
\newcommand{\ad}{\ensuremath{\mathrm{ad}}}
\newcommand{\g}{\ensuremath{\mathfrak{g}}}
\newcommand{\SEtt}{\mathsf{SE}_2(3)}
\newcommand{\sett}{\mathfrak{se}_2(3)}
\newcommand{\aSEtt}[3]{\begin{bmatrix} #1 & #2 & #3 \\ 0_{1\times 3} & 1 & 0 \\ 0_{1\times 3} & 0 & 1 \end{bmatrix}}
\newcommand{\asett}[3]{\begin{bmatrix} #1 & #2 & #3 \\ 0_{1\times 3} & 0 & 0 \\ 0_{1\times 3} & 0 & 0 \end{bmatrix}}
\newcommand{\pfimply}[2]{{(#1)$\,\Rightarrow\,$(#2)}}
\newcommand{\pfequiv}[2]{{(#1)$\,\Leftrightarrow\,$(#2)}}
\newcommand{\grav}{\mathsf{g}}
\newcommand{\notLeftrightarrow}{\mathrel{\ooalign{$\Leftrightarrow$\cr\hidewidth$/$\hidewidth}}}

\newtheorem{definition}{Definition}
\newtheorem{theorem}{Theorem}
\newtheorem{lemma}{Lemma}
\newtheorem{remark}{Remark}
\newtheorem{assumption}{Assumption}
\newtheorem{corollary}{Corollary}
\newtheorem{example}{Example}

\title{Invariant Kalman Filter for Relative Dynamics}

\author{Tejaswi K. C., Maneesha Wickramasuriya, Silv\`ere Bonnabel, Axel Barrau, and Taeyoung Lee
	\thanks{Tejaswi K. C., Maneesha Wickramasuriya, and Taeyoung Lee, Mechanical and Aerospace Engineering, George Washington University, Washington, DC 20052, {\tt \{kctejaswi999, maneesh, tylee\}@gwu.edu}}%
	\thanks{Axel Barrau, OFFROAD}%
	\thanks{Silv\`ere Bonnabel, Mines Paris, PSL University, Centre for robotics (CAOR), 75006 Paris, France }%
	\thanks{\textsuperscript{\footnotesize\ensuremath{*}}This research has been supported in part by AFOSR MURI FA9550-23-1-0400, and ONR N00014-23-1-2850.}
}

\begin{document}
\allowdisplaybreaks
\maketitle \thispagestyle{empty} \pagestyle{empty}

\begin{abstract}
This paper develops a geometric framework for invariant filtering of relative dynamics on Lie groups.  
We first revisit the notion of state trajectory independence, under which the estimation error evolves autonomously, and derive new equivalent conditions by decomposing the system vector field into left-invariant, intrinsic, and right-invariant components.  
Building on this result, we introduce the concept of relative trajectory independence to characterize when the relative motion between two dynamical systems is autonomous.  
A key theoretical finding is that relative trajectory independence automatically ensures state trajectory independence for the corresponding estimation error.  
This connection provides the foundation for constructing invariant filters that preserve the Lie group structure, maintain exact linearization of the error dynamics, and enable consistent covariance propagation.  
These are illustrated with numerical examples. 
\end{abstract}

\section{Introduction}

Although geometric methods have long been used for control~\cite{jurdjevic1972control}, their use for state estimation and observer design is relatively more recent~\cite{aghannan2002invariant,mahony2008nonlinear,bonnabel2008symmetry}.  
Since then, various methods have emerged, all revolving around the idea of defining an error variable through a Lie group operation.  
In particular, the invariant extended Kalman filter (IEKF)~\cite{barrau2016invariant} is a variation of the Kalman filter that seeks an appropriate group structure for the state space and leverages the Lie group error and exponential state update for estimation.  
When applied to certain classes of systems, it recovers several properties of the linear Kalman filter.  
Namely, it possesses convergence guarantees as an observer~\cite{barrau2016invariant,wang2020hybrid}, maintains consistency even in partially unobservable scenarios such as simultaneous localization and mapping~\cite{brossard2018exploiting,zhang2017convergence,heo2018consistent}, handles intrinsic constraints naturally~\cite{barrau2019extended,hartley2020iekf}, and has inspired various extensions~\cite{van2020invariant}.  
It has also led to high-end industrial products in inertial navigation and successful robotic applications~\cite{hartley2020iekf}.  

It has been shown that, for certain classes of dynamical systems, the estimation error—defined as the action of the true state’s inverse on the estimated state—evolves autonomously, implying that the error dynamics do not depend on the estimated state.  
This property is referred to as \emph{state trajectory independence (STI)} of the estimation error.  
This feature is particularly advantageous in extended Kalman filters, where linearized error dynamics are used to propagate the covariance of Gaussian distributions.  
In general, the error dynamics depend on the estimated state, and as such, the reliability of the propagated covariance can vary significantly with estimation accuracy.  
However, when applied to systems that exhibit STI, invariant Kalman filters achieve more consistent performance.  
They demonstrate robustness, as the accuracy of the linearization remains unaffected by the estimated state.  
Moreover, a linear differential equation underlies any nonlinear STI error equation—a property known as \emph{error log-linearity}~\cite{barrau2016invariant}—which plays a key role in proving IEKF convergence.  

This paper focuses on invariant Kalman filters for estimating the \emph{relative state} between two dynamical systems.
We consider a pair of systems that evolve under identical equations of motion but are driven by different control inputs.
Such problems arise in a wide range of engineering applications, including relative pose estimation~\cite{zhou2008robot}, spacecraft formation control~\cite{WuFleAJGCD17}, and autonomous rendezvous and docking~\cite{kelsey2006vision}.
The present study is particularly motivated by ongoing research on the autonomous flight of an unmanned aerial vehicle relative to a shipborne flight deck in a maritime environment~\cite{GamLeeCEP24,WicLeeSnyAJGCD25}, where accurate estimation of the vehicle’s pose relative to the flight deck is essential.

Owing to distinct control inputs, the relative discrepancy between two systems—expressed in terms of a group-invariant error—may not benefit from the cancellations that make single-system invariant errors trajectory-independent.  
However, in this paper, we show that an analogous independence property can be established for the relative estimation error between two dynamical systems.  
Specifically, while previous analyses of STI have been limited to single-system dynamics, the relative dynamics introduce coupling terms arising from distinct controls, making the autonomy of the relative estimation error nontrivial to characterize.  
To address this, we derive a new set of \emph{equivalent conditions} for STI by uniquely decomposing the system vector field into left-invariant, intrinsic, and right-invariant components.  
This decomposition provides a complete and geometrically transparent characterization of when the estimation error evolves independently of the underlying trajectory.

Building on this foundation, we introduce the concept of \emph{relative trajectory independence (RTI)}, which formalizes the autonomy of the relative dynamics between two systems.  
A key theoretical discovery of this work is that RTI \emph{automatically implies} STI for the corresponding relative estimation error.  
This result bridges the gap between single-system and relative-system formulations, showing that once the relative dynamics are trajectory-independent, the estimation error inherits this property without additional assumptions.  
This insight forms the theoretical cornerstone for the systematic construction of invariant filters for relative dynamics—an advancement not previously available in the literature.

Using these results, we develop left- and right-invariant relative Kalman filters that preserve the Lie group structure of the state space and maintain exact linearization of the error dynamics.  
These filters achieve consistent covariance propagation, robustness to large estimation errors, and symmetry-preserving performance across heterogeneous sensor configurations.  
The framework provides a unified geometric foundation for relative-state estimation applicable to multi-agent navigation, cooperative localization, and formation control.

The proposed theory is validated through numerical simulations using a vehicle kinematic model, where the proposed invariant filters are compared against a conventional quaternion-based extended Kalman filter (QEKF).  
The results demonstrate that the invariant filters consistently outperform the QEKF under various process and measurement noise conditions.

In short, the main contributions of this paper lie in the development of the theoretical foundation of invariant filters for relative dynamics.
The specific contributions include
\begin{enumerate}
    \item Derivation of new \emph{equivalent conditions} for state trajectory independence through a unique decomposition of the system vector field. 
    \item Introduction of the concept of \emph{relative trajectory independence} (RTI) and proof that RTI \emph{implies} STI for the relative estimation error.
    \item Construction of left- and right-invariant relative Kalman filters enabled by these theoretical results. 
    \item Numerical validation demonstrating that invariant filters achieve superior performance when their invariance type matches that of the measurement model.
\end{enumerate}

The remainder of this paper is organized as follows.  
\Cref{sec:STIRM} presents the theoretical development of state trajectory independence and its equivalent conditions.  
\Cref{sec:IKFRM} applies these results to construct invariant filters for relative dynamics, followed by numerical simulations in \Cref{sec:Num}, and concluding remarks in \Cref{sec:CL}.

\section{Invariant Kalman Filter}\label{sec:PF}

\subsection{Lie Group}

We first summarize the basic definitions and properties of a Lie group~\cite{marsden2013introduction}.  
A \textit{Lie group} $\G$ is a differentiable manifold endowed with a group structure such that the group operation is smooth.  
Throughout this paper, we consider $\G$ to be a matrix Lie group, i.e., $\G \subset \mathrm{GL}(n)$.  
The \textit{Lie algebra} $\g$ is defined as the tangent space of $\G$ at the identity element $e \in \G$, equipped with a \textit{Lie bracket} $[\cdot, \cdot] : \g \times \g \to \g$ that is bilinear, skew-symmetric, and satisfies the Jacobi identity.

For $g, h \in \G$, the \textit{left translation map} $\L_h : \G \to \G$ is defined by $\L_h(g) = h g$, and the \textit{right translation map} $\R_h : \G \to \G$ is defined by $\R_h(g) = g h$.  
Given $\xi \in \g$, define a vector field $X_\xi : \G \to \T \G$ by $X_\xi(g) = \T_e \L_g \cdot \xi$, and let $\gamma_\xi(t)$ denote the corresponding integral curve passing through the identity $e$ at $t = 0$.  
The \textit{exponential map} $\exp : \g \to \G$ is defined by $\exp(\xi) = \gamma_\xi(1)$.  
The map $\exp$ is a local diffeomorphism from a neighborhood of the origin in $\g$ onto a neighborhood of $e$ in $\G$.

An \textit{automorphism} is a map $\phi : \G \to \G$ such that $\phi(gh) = \phi(g)\phi(h)$.  
The \textit{inner automorphism} $\mathsf{I}_g : \G \to \G$ is defined by $\mathsf{I}_g(h) = g h g^{-1}$.  
The \textit{adjoint operator} $\Ad_g : \g \to \g$ is the differential of $\mathsf{I}_g$ at $h = e$ applied to $\eta \in \g$, i.e., $\Ad_g \eta = \T_e \mathsf{I}_g \cdot \eta$.  
The corresponding infinitesimal operator, denoted $\ad_\xi : \g \to \g$, is obtained by differentiating $\Ad_g \eta$ with respect to $g$ at $e$ in the direction $\xi$, that is,
\[
    \ad_\xi \eta = \T_e (\Ad_g \eta) \cdot \xi.
\]
This operator coincides with the Lie bracket: $\ad_\xi \eta = [\xi, \eta]$.

Let the dimension of $\G$ (and $\g$) be $n$.  
Since $\g$ is a vector space, we may choose a basis $(E_1, E_2, \ldots, E_n)$ for $\g$.  
Then, any $\xi \in \g$ can be expressed as $\xi = \sum_i \xi_i E_i$, where $\xi_i \in \Re$.  
This induces the \textit{vee map} $\vee : \g \to \Re^n$ defined by $\xi^\vee = (\xi_1, \xi_2, \ldots, \xi_n)$, and its inverse, the \textit{hat map} $\wedge : \Re^n \to \g$, given by $(\xi_1, \xi_2, \ldots, \xi_n)^\wedge = \xi$.

For the three-dimensional special orthogonal group $\SO3= \{ R \in \Re^{3 \times 3} \mid R^\top R = I_{3 \times 3}, \ \det[R] = 1 \}$, the hat map is defined so that $\hat{x} y = x \times y$ and $\hat{x}^\top = -\hat{x}$ for any $x, y \in \Re^3$.  
This defines a Lie algebra isomorphism since $\ad_{\hat{x}} \hat{y} = (x \times y)^\wedge$.  
The same notation for the hat and vee maps will be used for other Lie groups as well.

\subsection{State Trajectory Independence of Estimation Error (ETI)}

In this section, we revisit the development of invariant filtering in~\cite{barrau2016invariant} and present an alternative formulation that is particularly useful for relative dynamics. 
Consider a dynamical system on a Lie group $\G$, governed by the differential equation
\begin{align}
    \dot g = X(g, u), \label{eqn:dot_g_G}
\end{align}
where $u$ is a time-varying exogenous signal (control input) taking values in the input space $\U$, and $X : \G \times \U \rightarrow \T \G$ is a vector field dependent on the control.  
It is assumed that the existence and uniqueness of the solution of \eqref{eqn:dot_g_G} are guaranteed. 

The tangent bundle $\T \G$ may be trivialized either by left or right translation using the identification $\T \G \simeq \G \times \g$.  
Specifically, \eqref{eqn:dot_g_G} can be rewritten as
\begin{align}
    \dot g = g \eta(g,u), \quad \dot g = \xi(g, u) g, \label{eqn:dot_g}
\end{align}
where $\eta, \xi : \G \times \U \rightarrow \g$.  
Here, $g \eta$ denotes $\T_e \L_g \cdot \eta \in \T_g \G$, where $\T_e \L_g : \g \simeq \T_e \G \to \T_g \G$ is the tangent map of the left translation $\L_g$.  
For matrix Lie groups, $\T_e \L_g$ corresponds to left multiplication by $g$, making the notation $g\eta$ natural.  
Similarly, $\xi g$ corresponds to the right translation $\R_g$.  
In the first equation of \eqref{eqn:dot_g}, the transformation from $\dot g$ to $\eta$ is called the \textit{left trivialization}, and its inverse is the \textit{left translation}.  
The \textit{right trivialization} and \textit{right translation} are defined analogously.  
The left and right trivializations are related by $\xi = \Ad_g \eta$.  
In subsequent developments, we may drop the explicit dependence on $g$ or $u$ for brevity. 

Next, suppose there exists another trajectory $\bar g$ on $\G$, which evolves according to \eqref{eqn:dot_g} under the same control $u$, but potentially from a different initial condition.  
It may represent, for instance, the evolution of the estimated state of an observer:
\begin{align}
    \dot{\bar g} = X(\bar g, u). \label{eqn:bar_dot_g}
\end{align}

Define the \textit{left-invariant error} $f \in \G$ and the \textit{right-invariant error} $h \in \G$ between $g$ and $\bar g$ as
\begin{align}
    f = g^{-1} \bar g, \quad h = \bar g g^{-1}. \label{eqn:f}
\end{align}
Left invariance implies that if both $g$ and $\bar g$ are left-multiplied by another group element $l \in \G$, the error remains unchanged, i.e., $(l g)^{-1} (l \bar g) = g^{-1} \bar g = f$.  
Right invariance is defined analogously. 

One of the main contributions of~\cite{barrau2016invariant} is the formulation of the notion of \textit{state trajectory independence}, meaning that the dynamics of the error $f$ (or $h$) do not depend explicitly on $g$ or $\bar g$. 

\begin{definition}{(ETI)}\label{def:sti}
    Consider the trajectories $g$ of \eqref{eqn:dot_g} and $\bar g$ of \eqref{eqn:bar_dot_g} driven by the same control $u$.  
    Let the error state between $g$ and $\bar g$ be defined as in \eqref{eqn:f}.  
    We say that ``the \underline{error state} of \eqref{eqn:dot_g} is state trajectory independent,'' or equivalently, that ``\eqref{eqn:dot_g} is ETI,'' if the error is governed by a differential equation that depends only on the error state itself and the control. 
\end{definition}

This property is particularly useful for implementing extended Kalman filters, where uncertainties are propagated using Gaussian distributions over linearized dynamics.  
Since the linearization of the estimation error does not depend on the estimated state, the accuracy of uncertainty propagation is not degraded by poor initial estimates.

While a mathematical condition on $X$ equivalent to ETI was presented in~\cite{barrau2016invariant}, we introduce an alternative but equivalent condition that provides new insight and interpretation.  
Before doing so, we show that \eqref{eqn:dot_g_G} can be rearranged into equivalent forms as follows.

\begin{lemma}\label{lem:decomp}
    Consider \eqref{eqn:dot_g_G}, rewritten as
    \begin{align}
        \dot g = X_t(g). \label{eqn:dot_g_Xt}
    \end{align}
    Without loss of generality, it can be decomposed as
    \begin{align}
        \dot g = \xi(t) g + \tilde X_t(g) + g \zeta(t), \label{eqn:dot_g_Xtilde}
    \end{align}
    where $\tilde X_t : \Re \times \G \to \T \G$ is a time-dependent vector field satisfying
    \begin{align}
        \tilde X_t(e) = 0, \label{eqn:tilde_X_e}
    \end{align}
    and $\xi, \zeta : \Re \to \g$ are trajectories in $\g$ such that
    \begin{align}
        X_t(e) = \xi(t) + \zeta(t). \label{eqn:Xte}
    \end{align}
\end{lemma}
\begin{proof}
    Let $\tilde X_t(g) = X_t(g) - \xi g - g \zeta$.  
    Then $\tilde X_t(e) = X_t(e) - \xi - \zeta$, which vanishes by \eqref{eqn:Xte}. 
\end{proof}

This lemma states that any time-dependent vector field $X_t(g)$ on $\G$ can be decomposed as in \eqref{eqn:dot_g_Xtilde}, where the middle term vanishes when $g = e$.  
An interesting feature of \eqref{eqn:dot_g_Xtilde} is that $g$ appears only through left and right translations in the first and third terms, meaning the paths $\xi$ and $\zeta$ in $\g$ do not directly depend on $g$.  
The decomposition is not unique, since $\tilde X_t$, $\xi$, and $\zeta$ depend on how $X_t(e)$ is partitioned in \eqref{eqn:Xte}.  
For instance, one may choose $X_t(e) = \zeta(t)$ and $\xi(t) = 0$, or vice versa. 

\begin{example}\label{ex:UAV1}
    Consider a vehicle kinematic model with an onboard IMU, neglecting noise and bias.  
    Let $\mathcal{E}$ be the inertial frame whose third axis points downward along gravity, and let $\mathcal{B}$ be the body frame attached to the vehicle.  
    Define:
    \begin{itemize}
        \item $x \in \Re^3$: vehicle position relative to $\mathcal{E}$, resolved in $\mathcal{E}$.
        \item $v = \dot x \in \Re^3$: vehicle velocity relative to $\mathcal{E}$, resolved in $\mathcal{E}$.
        \item $R \in \SO3$: attitude of the vehicle, corresponding to the linear transformation of a representation of a vector from $\mathcal{B}$ to $\mathcal{E}$.
        \item $\Omega \in \Re^3$: angular velocity of the vehicle resolved in $\mathcal{B}$.
        \item $a = R^\top(\dot v - \grav e_3) \in \Re^3$: acceleration measured by the IMU, resolved in $\mathcal{B}$, where $e_3 = [0, 0, 1]^T$.
        \item $\grav \in \Re$: gravitational acceleration.
    \end{itemize}
    Assuming $\Omega$ and $a$ are measured by the IMU, the equations of motion are
    \begin{align*}
        \dot x & = v,\\
        \dot v & = R a + \grav e_3,\\
        \dot R & = R \hat\Omega.
    \end{align*}

    In~\cite{barrau2016invariant}, the states are embedded into the matrix Lie group $\SEtt$ (see Appendix~\ref{sec:SEtt}), such that the equations can be written as
    \begin{align}
        \dot g = g \eta,
    \end{align}
    where $g \in \SEtt$ and $\eta \in \sett$ are given by
    \begin{align*}
        g & = \aSEtt{R}{v}{x},\\
        \eta(g,u) & = \asett{\hat\Omega}{a + R^\top \grav e_3}{R^\top v},
    \end{align*}
    with $u = (\Omega, a)$.  
    For brevity, denote $g = (R, v, x)$ and $\eta(g,u) = (\hat\Omega, a + R^\top \grav e_3, R^\top v)$.  
    From \eqref{eqn:Xte}, we have $X_t(e) = \eta(e,u) = (\hat\Omega, a + \grav e_3, 0)$.

    According to \Cref{lem:decomp}, this can be decomposed as
    \begin{align}
        \xi(t) = (0, \grav e_3, 0), \quad \zeta(t) = (\hat\Omega, a, 0), \label{eqn:xi_zeta_IMU}
    \end{align}
    so that the equations of motion take the form of \eqref{eqn:dot_g_Xtilde} with
    \begin{align}
        \tilde X(g) = (0, 0, v). \label{eqn:tilde_X_3_IMU}
    \end{align}
    This decomposition is physically natural: the gravitational acceleration $\grav e_3$ is resolved in $\mathcal{E}$ and right-translated by $g$, while $\Omega$ and $a$ are resolved in $\mathcal{B}$ and left-translated by $g$ to be expressed in $\mathcal{E}$.  
    The remaining vector field $\tilde X(g)$ in \eqref{eqn:tilde_X_3_IMU} no longer depends on $t$. 
\end{example}
This example will be used repeatedly throughout this paper, to illustrate the presented theoretical results more clearly. 

The equivalent condition for the state trajectory independence of the error state in \eqref{eqn:dot_g_Xt} was presented in~\cite{barrau2016invariant}.  
Here, we develop a similar result for the alternative formulation given by \eqref{eqn:dot_g_Xtilde}.

\begin{theorem}{(ETI)}\label{thm:groupaff}
    Consider the dynamics on $\G$ given by \eqref{eqn:dot_g_Xt} or \eqref{eqn:dot_g_Xtilde}.
    The following statements are equivalent. 

    \renewcommand{\labelenumi}{\textit{(\roman{enumi})}}
    \begin{enumerate} 
        \item The left-invariant error and the right-invariant error are state trajectory independent (ETI).
        \item For any $g , h \in \G $ and $t\in\Re$,  the vector field $X_t$ of \eqref{eqn:dot_g_Xt} satisfies
            \begin{align}
                X_t(hg) = h X_t(g) + X_t(h)g - h X_t(e)g. \label{eqn:group_aff}
            \end{align}
        \item For any $g , h \in \G $ and $t\in\Re$,  the vector field $\tilde X_t$ of \eqref{eqn:dot_g_Xtilde} satisfies
            \begin{align}
                \tilde X_t(hg)=h\tilde X_t(g)+\tilde X_t(h)g. \label{eqn:group_aff_tilde}
            \end{align}
    \end{enumerate}
\end{theorem}
\begin{proof}
    It has been shown that \pfequiv{i}{ii} in \cite{barrau2016invariant}.
    Next, suppose (iii) holds.
    By equating \eqref{eqn:dot_g_Xt} and \eqref{eqn:dot_g_Xtilde}, 
    \begin{align}
        X_t(g) = \tilde X_t(g) + \xi g + g \zeta.\label{eqn:X_tildeX}
    \end{align}
    Setting $g=hg$ in \eqref{eqn:X_tildeX}, and substituting \eqref{eqn:group_aff_tilde}, we have
    \begin{align*}
        X_t(hg) = h\tilde X_t (g) + \tilde X_t (h)g + \xi hg + hg \zeta.
    \end{align*}
    Using \eqref{eqn:X_tildeX} repeatedly, this is rearranged into
    \begin{align*}
        X_t(hg) & = h(X_t(g) - \xi g -g\zeta)\\
                & \quad+ ( X_t (h)-\xi h -h\zeta) g + \xi hg + hg \zeta\\
                & = hX_t(g) + X_t(h)g  - h(\xi+\zeta)g
    \end{align*}
    But, by letting $g=e$ in \eqref{eqn:X_tildeX}, we have $X_t(e) = \xi + \zeta$, 
    which shows that the above implies \eqref{eqn:group_aff}.
    Thus, \pfimply{iii}{ii}.

    Next, suppose (ii) is true. From \eqref{eqn:group_aff} and \eqref{eqn:X_tildeX}, 
    \begin{align*}
        \tilde X_t (hg) = h X_t(g) + X_t(h)g - h X_t(e)g - \xi hg - hg \zeta.
    \end{align*}
    Substituting \eqref{eqn:X_tildeX} repeatedly, and using $X_t(e) = \xi + \zeta$, this reduces to \eqref{eqn:group_aff_tilde}.
    Therefore, \pfimply{ii}{iii}.
\end{proof}

This lemma states that once the vector field at $g=e$, or $X_t(e)$ is extracted out (after being split and translated), the state trajectory independence of error can be determined solely by the remaining part $\tilde X_t(g)$ via \eqref{eqn:group_aff_tilde}.
In other words, the state trajectory of the error state of \eqref{eqn:dot_g_Xtilde} is not affected by the presence of $\xi$ and $\zeta$ that are independent of $g$. 

\begin{example}{(ETI)}\label{ex:ETI}
    We verify the state trajectory independence of the error for the vehicle dynamics presented in \Cref{ex:UAV1}.
    Let $g=(R,v,x),h=(T, p, q)\in\SEtt$. 
    Then, we have $hg = (TR, Tv+p, Tx+q)$. 
    From \eqref{eqn:tilde_X_3_IMU}, we have $\tilde X(hg) = (0, 0, Tv+p)$. 
    On the other hand, 
    \begin{align*}
        &h \tilde X(g) +\tilde X(h) g =
        \aSEtt{T}{p}{q}\asett{0_{3\times 3}}{0_{3\times 1}}{v} \\
        &+ \asett{0_{3\times 3}}{0_{3\times 1}}{p} \aSEtt{R}{v}{x}\\
        & = \asett{0_{3\times 3}}{0_{3\times 1}}{Tv+p}.
    \end{align*}
    Therefore, \eqref{eqn:group_aff_tilde} is satisfied, and according to \Cref{thm:groupaff}, the error state of the vehicle dynamics is state trajectory independent. 

    While we verified \eqref{eqn:group_aff_tilde} using the specific choice of $\tilde X$ presented in \Cref{ex:UAV1}, one can show \eqref{eqn:group_aff_tilde} is satisfied by any of choice of $\xi$, $\zeta$, and $\tilde X$.
\end{example}

The following corollary presents the governing equations for the error state when the errors are state trajectory independent. 

\begin{corollary}{(ETI)}\label{cor:dot_f_h}
    Suppose \eqref{eqn:dot_g_Xt} or \eqref{eqn:dot_g_Xtilde} is state trajectory independent. 
    The corresponding left-invariant error $f=g^{-1}\bar g$ is governed by
    \begin{align}
        \dot f & = X_t(f) - X_t(e) f,\label{eqn:dot_f_0}\\
               & =\tilde X_t(f) + f \zeta(t) - \zeta(t) f,\label{eqn:dot_f_1}
    \end{align}
    Similarly, the right-invariant error $h = \bar g g^{-1}$ is governed by
    \begin{align}
        \dot h & = X_t(h) - h X_t(e),\label{eqn:dot_h_0} \\
               & =\tilde X_t(h) + \xi(t) h - h \xi(t),\label{eqn:dot_h_2}
    \end{align}
\end{corollary}
\begin{proof}
    First, we show \eqref{eqn:dot_f_1} as follows.  
    From \eqref{eqn:group_aff_tilde}, we have
    $ \tilde X_t(g f) = g \tilde X_t(f) + \tilde X_t(g)f$,
    which is rearranged into the following, after multiplying both sides by $g^{-1}$ on the left with $gf = \bar g$:
    \begin{align}
        g^{-1}\tilde X_t(\bar g) - g^{-1} \tilde X_t(g)f =  \tilde X_t(f) .
        \label{eqn:cor1_0}
    \end{align}
    The time-derivative of $f=g^{-1}\bar g$ is given by
    \begin{align*}
        \dot f & = -g^{-1} \dot g g^{-1} \bar g + g^{-1}\dot {\bar g}\\
               & = -g^{-1} (\xi g + \tilde X_t(g) + g\zeta)  f + g^{-1}(\xi \bar g + \tilde X_t(\bar g) + \bar g \zeta).
    \end{align*}
    Substituting \eqref{eqn:cor1_0} and $gf=\bar g$, this reduces to \eqref{eqn:dot_f_1}.
    Next, \eqref{eqn:dot_f_0} is presented in \cite{barrau2016invariant}, or it can be obtained by substituting $X_t(f) = \xi f + \tilde X_t(f) + f\zeta = (X_t(e) -\zeta)f + \tilde X_t(f) + f\zeta$ into \eqref{eqn:dot_f_1}.

    Finally, \eqref{eqn:dot_h_0} and \eqref{eqn:dot_h_2} can be derived similarly. 
\end{proof}

While \eqref{eqn:dot_f_0}--\eqref{eqn:dot_h_2} are stated as the necessary conditions of the state trajectory independence (ETI) in \Cref{cor:dot_f_h}, they are also sufficient conditions. 
More specifically, if \eqref{eqn:dot_f_0} is satisfied, then it directly implies that the error is governed by a differential equation dependent of the error itself, and therefore, the error of \eqref{eqn:dot_g_Xt} or \eqref{eqn:dot_g_Xtilde} is state trajectory independent according to \Cref{def:sti}.
As such, each of \eqref{eqn:dot_f_0}--\eqref{eqn:dot_h_2} is an equivalent condition to \Cref{thm:groupaff}.

Next, interestingly, as shown in \eqref{eqn:dot_f_1}, the dynamics of the left-invariant error $f$ does not depend on $\xi$ of \eqref{eqn:dot_g_Xtilde}.
Similarly, in \eqref{eqn:dot_h_2}, the dynamics of the right-invariant error $h$ is not contributed by $\zeta$ of \eqref{eqn:dot_g_Xtilde}.
This is further illustrated by the following example. 

\begin{example}
    Let $g = (R, v, x)$ and $\bar g = (\bar R, \bar v, \bar x)$ such that the left-invariant error $f=g^{-1}\bar g$ is obtained by
    \begin{align*}
        f & = \aSEtt{R^T}{-R^Tv}{-R^T x}\aSEtt{\bar R}{\bar v}{\bar x}\\
          & = \aSEtt{R^T \bar R}{R^T(\bar v -v)}{R^T(\bar x -x)},
    \end{align*}
    which is defined to be $f = (Q, u, y)$. 
    From \Cref{ex:UAV1}, we have $\tilde X(f) = (0,0, u)$ and $\zeta = (\hat \Omega , a, 0)$. 
    Substituting this into \eqref{eqn:dot_f_1}, 
    \begin{align*}
        \dot f & = (0, 0, u) +
        \aSEtt{Q}{u}{y}\asett{\hat\Omega}{a}{0} \\
               & \quad - \asett{\hat\Omega}{a}{0}\aSEtt{Q}{u}{y}\\
               & = (0,0,u) + (Q\hat\Omega, Qa, 0) - (\hat\Omega Q, \hat\Omega u +a, \hat\Omega y)\\
               & = (Q\hat\Omega - \hat\Omega Q, Qa -a -\hat\Omega u , u-\hat\Omega y),
    \end{align*}
    which is independent of $\xi = (0, \grav e_3, 0)$. 

    Similarly, let $h = (T, p, q)\in\SEtt$. 
    From \eqref{eqn:dot_h_2}, we can show that
    \begin{align*}
        \dot h = (0, \grav e_3 - T\grav e_3, p),
    \end{align*}
    which does not depend on $\zeta = (\hat\Omega, a, 0)$.
\end{example}

\begin{remark}{(Geometric Interpretation)}\label{rem:ETI_geo}
    The decomposition \eqref{eqn:dot_g_Xtilde} reveals a clear geometric structure underlying state trajectory independence.  
    The terms $g \zeta(t)$ and $\xi(t) g$ correspond to infinitesimal left and right translations on the group, respectively, generated by the Lie algebra elements $\zeta(t)$ and $\xi(t)$ that evolve independently of the current state $g$.  
    These components represent uniform ``drift'' directions that are globally consistent across $\G$, and hence do not influence the relative motion between trajectories.  
    In contrast, the residual field $\tilde X_t(g)$ captures the part of the dynamics that depends on the group element itself, describing the geometry-dependent flow on $\G$.  
    The ETI property ensures that, when comparing two trajectories evolving under the same control input, the relative motion depends only on this intrinsic geometric component $\tilde X_t(g)$, and not on the global drift terms $\xi$ or $\zeta$.  
    This interpretation aligns with the fundamental idea of invariant filtering: estimation errors evolve according to the intrinsic group geometry, independent of the absolute state trajectory.
\end{remark}

\section{State Trajectory Independence for Relative Motion}\label{sec:STIRM}

In this section, we discuss how the concept of state trajectory independence can be extended to the estimation of relative motion between two trajectories of a dynamical system. 
This framework is particularly useful for the relative pose estimation of multi-agent systems, where each agent is governed by the same equation of motion but may be driven by distinct control inputs.

\subsection{Relative Motion Estimation Problem}

Consider two trajectories $g_1$ and $g_2$ of \eqref{eqn:dot_g_Xt} controlled by the inputs $u_1$ and $u_2$, respectively. 
Specifically,
\begin{align}
    \dot g_i = X(g_i, u_i(t)) = X_t^i(g_i),\label{eqn:dot_gi}
\end{align}
for $i\in\{1,2\}$. 
According to \Cref{lem:decomp}, this can be decomposed as
\begin{align}
    \dot g_i = \xi_i(t) g_i + \tilde X^i_t(g_i) + g_i\zeta_i(t), \label{eqn:dot_gi_Xtilde}
\end{align}
where $\tilde X^i:\Re\times\G\to\T\G$ and $\xi_i,\zeta_i:\Re\to\g$ for $i\in\{1,2\}$ satisfy
\begin{align}
    X(e, u_i) = \xi_i + \zeta_i,\quad \tilde X(e, u_i) = 0.
\end{align}

The relative motion between $g_1$ and $g_2$ can be described by the left-invariant relative configuration
\begin{align}
    g_{12} = g_1^{-1} g_2,\label{eqn:g12_L}
\end{align}
or by the right-invariant relative configuration
\begin{align}
    g_{12} = g_1 g_2^{-1},\label{eqn:g12_R}
\end{align}
where we adopt a slight abuse of notation by denoting both relative states by $g_{12}$.
They will, however, be clearly distinguished in the subsequent developments. 

The objective of this section is to formulate the state trajectory independence of these relative dynamics to construct an invariant filter for the relative state $g_{12}$ using the concatenated control $u_{12}=(u_1, u_2)$.

\subsection{State Trajectory Independence of Relative Motion (RTI)}

The first question to address is whether we can construct an estimator solely in terms of the relative state $g_{12}$. 
In general, estimating the relative motion requires knowledge of both $g_1$ and $g_2$. 
In other words, the relative-state estimator must be constructed for $(g_1,g_2)$ (or equivalently for $(g_1, g_{12})$ or $(g_{12}, g_2)$) on $\G\times\G$. 
This is undesirable since it doubles the dimension of the state to be estimated. 

However, if the dynamics of the relative motion depend on neither $g_1$ nor $g_2$, the estimator for the relative motion $g_{12}$ can be constructed directly on $\G$ without the need to estimate $g_1$ or $g_2$.  
This property is defined as follows. 

\begin{definition}{(RTI)}\label{def:rel_sti}
    Consider the trajectories $g_1$ and $g_2$ of \eqref{eqn:dot_gi} driven by $u_1$ and $u_2$, respectively.
    Let the relative motion between $g_1$ and $g_2$ be defined by \eqref{eqn:g12_L} or \eqref{eqn:g12_R}.
    We state that ``the left (resp. right) \underline{relative state} of \eqref{eqn:dot_gi} is state trajectory independent,'' or that ``\eqref{eqn:dot_gi} is $L$-RTI (resp. $R$-RTI),'' if the relative state $g_{12}=g_1^{-1}g_2$ (resp. $g_{12}= g_1 g_2^{-1}$) is governed by a differential equation dependent only on the relative state itself and the concatenated control $u_{12}=(u_1, u_2)$ only. 
\end{definition}

While this is analogous to the state trajectory independence of the \underline{error state} (ETI) defined in \Cref{def:sti}, the state trajectory independence of the \underline{relative motion} (RTI) deals with two trajectories influenced by two distinct control inputs. 
More explicitly, \Cref{def:sti} concerns the error between $g$ and $\bar g$ driven by the same input $u$, whereas \Cref{def:rel_sti} deals with the difference between $g_1$ and $g_2$ driven by $u_1$ and $u_2$, respectively. 
Thus, we distinguish the state trajectory independence of the \underline{estimation error} from that of the \underline{relative trajectory}. 

The equivalent conditions for the state trajectory independence of the relative state are summarized below, first for the left-invariant relative state ($L$-RTI), and later for the right-invariant relative state ($R$-RTI). 

\begin{theorem}{($L$-RTI)}\label{thm:rel_sti_L}
    Consider the left-invariant relative trajectory $g_{12}=g_1^{-1}g_2$ given by \eqref{eqn:g12_L}.
    The following statements are equivalent. 

    \renewcommand{\labelenumi}{\textit{(\roman{enumi})}}
    \begin{enumerate}
        \item The left-invariant relative state $g_{12}$ of \eqref{eqn:g12_L} is state trajectory independent ($L$-RTI). 
        \item For any $h,g\in\G$ and $t\in\Re$, the vector field $X$ in \eqref{eqn:dot_gi} satisfies
            \begin{align}
                X(hg, u_2)= hX(g, u_2) + X(h, u_1)g - hX(e, u_1)g .\label{eqn:cond_g12L_autonomous}
            \end{align}
        \item For any $h,g\in\G$ and $t\in\Re$, the vector field $\tilde X$ in \eqref{eqn:dot_gi_Xtilde} satisfies
            \begin{align}
                \tilde X(hg, u_2) & = h\tilde X(g,u_2) + \tilde X(h, u_1)g  \nonumber \\
                                  & \quad + h (\xi_2-\xi_1) g + (\xi_1-\xi_2)hg .\label{eqn:cond_g12Ltilde_autonomous}
            \end{align}
        \item The relative state is governed by
            \begin{align}
                \dot g_{12} & = X(g_{12}, u_2) - X(e, u_1)g_{12}\label{eqn:dot_g12_L}, \\
                            & = (\xi_2-\xi_1-\zeta_1) g_{12} +\tilde X(g_{12}, u_2) +  g_{12}\zeta_2 . \label{eqn:dot_g12_L_tilde}
            \end{align}
    \end{enumerate}
\end{theorem}

\begin{proof}
    It is straightforward to show \pfequiv{ii}{iii} using $X^i(g_i) = \xi_i g_i + \tilde X^i(g_i) + g_i\zeta_i$.
    Similarly, \eqref{eqn:dot_g12_L} is equivalent to \eqref{eqn:dot_g12_L_tilde}.

    Next, if (ii) holds, from \eqref{eqn:cond_g12L_autonomous}, we have
    \begin{align}
        X(g_1g_{12}, u_2) & = g_1X( g_{12}, u_2) + X(g_1, u_1) g_{12} \nonumber \\
                          & \quad - g_1 X(e, u_1)g_{12}.\label{eqn:thm2_0}
    \end{align}
    Multiplying both sides by $g_1^{-1}$ on the left and rearranging, 
    \begin{align}
        -g_1^{-1}& X(g_1, u_1) g_{12} + g_1^{-1} X(g_1g_{12},  u_2)\nonumber\\
                 &= X(g_{12}, u_2) - X(e,u_1)g_{12}.\label{eqn:thm2_1}
    \end{align}
    The time derivative of $g_{12} = g_1^{-1} g_2$ is given by
    \begin{align}
        \dot g_{12} & = -g_1^{-1} X(g_1, u_1) g_1^{-1}g_2 + g_1^{-1} X(g_2, u_2).\label{eqn:thm2_2}
    \end{align}
    Substituting \eqref{eqn:thm2_1} into \eqref{eqn:thm2_2} yields \eqref{eqn:dot_g12_L}.
    Hence, \pfimply{ii}{iv}.

    Suppose (iv) holds. 
    Then, \eqref{eqn:dot_g12_L} implies \eqref{eqn:thm2_1}, which can be rearranged into \eqref{eqn:thm2_0}, showing \eqref{eqn:cond_g12L_autonomous}. 
    Thus, \pfimply{iv}{ii}.

    If (i) is true, then the right-hand side of \eqref{eqn:thm2_2} should remain unchanged when $g_1$ and $g_2$ are multiplied by any $l\in\G$ on the left:
    \begin{align*}
        -g_1^{-1}& X(g_1, u_1) g_{12} + g_1^{-1} X(g_2, u_2) = \\
                 & -g_1^{-1}l^{-1} X(lg_1, u_1) g_{12} + g_1^{-1} l^{-1} X(lg_2, u_2).
    \end{align*}
    Setting $l=g_1^{-1}$ and left-multiplying both sides by $g_1$ recovers \eqref{eqn:thm2_0}, proving \pfimply{i}{ii}.

    Finally, \eqref{eqn:dot_g12_L} implies (i) by \Cref{def:rel_sti}, completing the equivalence of (i)–(iv).
\end{proof}

While the conditions \eqref{eqn:cond_g12L_autonomous} and \eqref{eqn:cond_g12Ltilde_autonomous} have a structure similar to \eqref{eqn:group_aff} and \eqref{eqn:group_aff_tilde}, they are more general since the state trajectory independence of the relative state deals with cases where $u_1$ and $u_2$ may differ. 
It is straightforward to observe that when $u_1=u_2$ and $\xi_1=\xi_2$, \eqref{eqn:cond_g12L_autonomous} and \eqref{eqn:cond_g12Ltilde_autonomous} reduce to \eqref{eqn:group_aff} and \eqref{eqn:group_aff_tilde}, respectively.

Next, we verify the $L$-RTI property for the IMU model below.

\begin{example}{($L$-RTI)}\label{ex:L-RTI}
    Let $g_1 = (R_1, v_1, x_1)$ and $g_2 = (R_2, v_2, x_2)$ so that the left-invariant relative state $g_{12}=g_1^{-1} g_2$ is given by
    \begin{align*}
        g_{12} & = (R_1^T R_2,\, R_1^T(v_2 - v_1),\, R_1^T(x_2 -x_1)),
    \end{align*}
    which is denoted by $g_{12} = (R_{12}, v_{12}, x_{12})$. 
    From \Cref{ex:UAV1}, we have $\tilde X(g_{12}) = (0,0, v_{12})$, $\xi_i = (0, \grav e_3, 0)$, and $\zeta_i = (\hat \Omega_i , a_i, 0)$. 

    We now verify \eqref{eqn:cond_g12Ltilde_autonomous}.
    Let $g=(R,v,x)$ and $h=(T, p, q)\in\SEtt$. 
    Then, $hg = (TR, Tv+p, Tx+q)$, implying that the left-hand side of \eqref{eqn:cond_g12Ltilde_autonomous} is $\tilde X(hg, u_2) = (0, 0, Tv+p)$. 
    On the other hand, since $\xi_1=\xi_2$,
    \begin{align*}
    & h\tilde X(g,u_2) + \tilde X(h, u_1)g  + h (\xi_2-\xi_1) g + (\xi_1-\xi_2)hg \\
    & = (T, p, q) (0, 0, v) + (0, 0, p) (R, v, x) \\
    & = (0, 0, Tv) + (0, 0, p) = (0, 0, Tv + p).
    \end{align*}
    Therefore, \eqref{eqn:cond_g12Ltilde_autonomous} is satisfied, and according to \Cref{thm:rel_sti_L}, the left-invariant relative state $g_{12}$ is state trajectory independent.

    Substituting this into \eqref{eqn:dot_g12_L_tilde}, we obtain
    \begin{align}
        \dot g_{12} & = - (\hat\Omega_1, a_1, 0) (R_{12}, v_{12}, x_{12}) + (0, 0, v_{12})\nonumber\\
                    & \quad + (R_{12}, v_{12}, x_{12}) (\hat\Omega_2, a_2, 0) \nonumber \\
                    & = (R_{12}\hat\Omega_2 - \hat\Omega_1 R_{12},\, R_{12}a_2 -\hat\Omega_1 v_{12} -a_1,\, v_{12} -\hat\Omega_1 x_{12}). \label{eqn:dot_g12_L_ex}
    \end{align}
\end{example}

Interestingly, according to \eqref{eqn:cond_g12Ltilde_autonomous}, the state trajectory independence of the left-invariant relative state remains unaffected by $\zeta_1$ and $\zeta_2$. 
Furthermore, \eqref{eqn:cond_g12Ltilde_autonomous} simplifies if $\xi_1\equiv\xi_2$, as illustrated in \Cref{ex:L-RTI}.

We can develop a similar result for the right-invariant relative state.

\begin{theorem}{($R$-RTI)}\label{thm:rel_sti_R}
    Consider the right-invariant relative trajectory $g_{12}=g_1g_2^{-1}$ given by \eqref{eqn:g12_R}.
    The following statements are equivalent. 

    \renewcommand{\labelenumi}{\textit{(\roman{enumi})}}
    \begin{enumerate}
        \item The right-invariant relative state $g_{12}$ of \eqref{eqn:g12_R} is state trajectory independent ($R$-RTI). 
        \item For any $h,g\in\G$ and $t\in\Re$, the vector field $X$ in \eqref{eqn:dot_gi} satisfies
            \begin{align}
                X(hg, u_1)= hX(g, u_2) + X(h, u_1)g - hX(e, u_2)g .\label{eqn:cond_g12R_autonomous}
            \end{align}
        \item For any $h,g\in\G$ and $t\in\Re$, the vector field $\tilde X$ in \eqref{eqn:dot_gi_Xtilde} satisfies
            \begin{align}
                \tilde X(hg, u_1) & = h\tilde X(g,u_2) + \tilde X(h, u_1)g  \nonumber \\
                                  & \quad + h (\zeta_1-\zeta_2) g + hg(\zeta_2-\zeta_1) .\label{eqn:cond_g12Rtilde_autonomous}
            \end{align}
        \item The relative state is governed by
            \begin{align}
                \dot g_{12} & = X(g_{12}, u_1) - g_{12} X(e, u_2),\label{eqn:dot_g12_R}\\
                            & = \xi_1 g_{12} + \tilde X(g_{12}, u_1) + g_{12}(\zeta_1-\zeta_2 - \xi_2 ) .\label{eqn:dot_g12_R_tilde}
            \end{align}
    \end{enumerate}
\end{theorem}

\begin{proof}
    It is straightforward to show \pfequiv{ii}{iii} using $X^i(g_i) = \xi_i g_i + \tilde X^i(g_i) + g_i\zeta_i$.
    Similarly, \eqref{eqn:dot_g12_R} is equivalent to \eqref{eqn:dot_g12_R_tilde}.

    If (ii) holds, from \eqref{eqn:cond_g12R_autonomous} we have
    \begin{align}
        X(g_{12}g_2, u_1)  &= g_{12} X(g_2,  u_2) + X(g_{12}, u_1)g_2 \nonumber\\
                           & \quad -g_{12} X(e,u_2)g_2.\label{eqn:thm3_0}
    \end{align}
    Multiplying both sides by $g_2^{-1}$ on the right and rearranging, 
    \begin{align}
        & X(g_{12}g_2, u_1)g_2^{-1}  -g_{12} X(g_2,  u_2)g_2^{-1}\nonumber\\
        &= X(g_{12}, u_1) - g_{12} X(e,u_2).\label{eqn:thm3_1}
    \end{align}
    The time derivative of $g_{12} = g_1 g_2^{-1}$ is given by
    \begin{align}
        \dot g_{12} & = X(g_1, u_1)g_2^{-1} - g_{12} X(g_2, u_2) g_2^{-1}.\label{eqn:thm3_2}
    \end{align}
    Substituting \eqref{eqn:thm3_1} into \eqref{eqn:thm3_2} with $g_1 =g_{12}g_2$ yields \eqref{eqn:dot_g12_R}.
    Therefore, \pfimply{ii}{iv}.

    Next, suppose (iv) holds. 
    Then, \eqref{eqn:dot_g12_R} implies \eqref{eqn:thm3_1}, which can be rearranged into \eqref{eqn:thm3_0}, showing \eqref{eqn:cond_g12R_autonomous}. 
    Thus, \pfimply{iv}{ii}.

    If (i) holds, the right-hand side of \eqref{eqn:thm3_2} must remain unchanged when $g_1$ and $g_2$ are multiplied by any $l\in\G$ on the right.
    Equivalently,
    \begin{align*}
        X&(g_1, u_1) g_2^{-1} - g_{12} X(g_2, u_2) g_2^{-1} \\
         &= X(g_1l , u_1) l^{-1} g_2^{-1} - g_{12} X(g_2l , u_2) l^{-1}g_2^{-1}.
    \end{align*}
    Setting $l=g_2^{-1}$ and right-multiplying both sides by $g_2$, with $g_1=g_{12}g_2$, recovers \eqref{eqn:thm3_0}, proving \pfimply{i}{ii}.

    Finally, \eqref{eqn:dot_g12_R} implies (i) by \Cref{def:rel_sti}, completing the equivalence of (i)–(iv).
\end{proof}

It turns out that the right-invariant relative state of the IMU model is not state trajectory independent, as shown below.

\begin{example}{($R$-RTI)}\label{ex:R-RTI}
    Let $g_1 = (R_1, v_1, x_1)$ and $g_2 = (R_2, v_2, x_2)$ such that the right-invariant error $g_{12}=g_1 g_2^{-1}$ is given by
    \begin{align*}
        g_{12} & = (R_1 R_2^T,\, -R_1R_2^Tv_2 + v_1,\, -R_1R_2^Tx_2 + x_1),
    \end{align*}
    denoted by $g_{12} = (R_{12}, v_{12}, x_{12})$. 
    From \Cref{ex:UAV1}, we have $\tilde X(g_{12}) = (0,0, v_{12})$, $\xi_i = (0, \grav e_3, 0)$, and $\zeta_i = (\hat \Omega_i, a_i, 0)$. 

    Let $g=(R,v,x)$ and $h=(T, p, q)\in\SEtt$. 
    The left-hand side of \eqref{eqn:cond_g12Rtilde_autonomous} is $\tilde X(hg, u_2) = (0, 0, Tv+p)$. 
    On the other hand,
    \begin{align*}
    & h\tilde X(g,u_2) + \tilde X(h, u_1)g  + h (\zeta_1-\zeta_2) g + hg(\zeta_2-\zeta_1)\\
    & = (T, p, q) (0, 0, v) + (0, 0, p) (R, v, x)  \\
    & \quad + (T, p, q) (\hat\Omega_1 - \hat\Omega_2, a_1 - a_2, 0)(R, v, x) \\
    & \quad - (TR, Tv+p, Tx+q)(\hat\Omega_1 - \hat\Omega_2, a_1 - a_2, 0) \\
    & \neq (0, 0, Tv + p),
    \end{align*}
    showing that \eqref{eqn:cond_g12Rtilde_autonomous} is not satisfied. 
    Therefore, according to \Cref{thm:rel_sti_R}, the right-invariant relative state $g_{12}$ is not state trajectory independent.
\end{example}

In \Cref{thm:groupaff}, we observed that the left-invariant error is state trajectory independent, if and only if the right-invariant error is state trajectory independent, i.e., $L$-ETI $\Leftrightarrow$ $R$-ETI.
However, as shown in \Cref{thm:rel_sti_L} and \Cref{thm:rel_sti_R} (or \Cref{ex:L-RTI} and \Cref{ex:R-RTI}), the state trajectory independence of the left-invariant relative state does not imply that of the right-invariant relative state, i.e., $L$-RTI $\notLeftrightarrow$ $R$-RTI. 

Next, as discussed above, it is straightforward to observe that \eqref{eqn:cond_g12L_autonomous} or \eqref{eqn:cond_g12R_autonomous} implies \eqref{eqn:group_aff}.
Therefore, the state trajectory independence of the relative state is a sufficient condition for the state trajectory independence of the estimation error, i.e., RTI $\Rightarrow$ ETI.

\begin{corollary}\label{cor:RTI_ETI}
    If the left-invariant or the right-invariant relative state $g_{12}$ of \eqref{eqn:dot_g} is state trajectory independent (RTI) as presented in \Cref{def:rel_sti}, then the estimation error $f$ and $h$ of \eqref{eqn:dot_g} is state trajectory independent (ETI) as formulated in \Cref{def:sti}. 
\end{corollary}
\begin{proof}
    Suppose the left-invariant (resp. right-invariant) relative state $g_{12}$ is state trajectory independent.
    According to \Cref{thm:rel_sti_L} (resp. \Cref{thm:rel_sti_R}), \eqref{eqn:cond_g12L_autonomous} (resp. \eqref{eqn:cond_g12R_autonomous}) holds for any $u_1$ and $u_2$. 
    Setting $u_1=u_2$ yields \eqref{eqn:group_aff}, and therefore $f$ and $h$ is state trajectory independent from \Cref{thm:groupaff}.
\end{proof}

The next natural question is: under what conditions does the state trajectory independence of error imply the state trajectory independence of relative states, i.e., ETI $\overset{?}{\Rightarrow}$ RTI.

\begin{corollary}\label{cor:ETI_RTI}
    Suppose that the estimation error of \eqref{eqn:dot_gi} is state trajectory independent (ETI). 
    Further, assume that $\tilde X_t^i(g)$ in \eqref{eqn:dot_gi_Xtilde} does not depend on $t$ and $i$, i.e., we can effectively rewrite $\tilde X^i_t(g) = \tilde X(g)$ for any $g\in\G$.
    Then, the following statements hold.
    \renewcommand{\labelenumi}{\textit{(\roman{enumi})}}
    \begin{enumerate}
        \item The left-invariant relative state $g_{12}=g_1^{-1}g_2$ is state trajectory independent ($L$-RTI), if and only if $\xi_1(t) \equiv \xi_2(t)$. 
            In this case, the relative state is governed by
            \begin{align}\label{eqn:dot_g12_L_cor}
                \dot g_{12} = -\zeta_1 g_{12} + \tilde X(g_{12}) + g_{12}\zeta_2.
            \end{align}
        \item The right-invariant relative state $g_{12}=g_1g_2^{-1}$ is state trajectory independent ($R$-RTI), if and only if $\zeta_1(t) \equiv \zeta_2(t)$. 
            In this case, the relative state is governed by
            \begin{align}\label{eqn:dot_g12_R_cor}
                \dot g_{12} = \xi_1 g_{12} + \tilde X(g_{12}) - g_{12}\xi_2.
            \end{align}
    \end{enumerate}
\end{corollary}
\begin{proof}
    Since it is ETI, according to \Cref{thm:groupaff}, \eqref{eqn:group_aff_tilde} is satisfied. 
    But, as $\tilde X_t^i(g)=\tilde X(g)$, it reduces to
    \begin{align}
        \tilde X(hg) = h \tilde X(g) + \tilde X(h) g. \label{eqn:cor3_0}
    \end{align}

    First, if $\xi_1\equiv \xi_2$, then \eqref{eqn:cond_g12Ltilde_autonomous} is rearranged into \eqref{eqn:cor3_0}, and according to \Cref{thm:rel_sti_L}, \eqref{eqn:dot_gi} is $L$-RTI. 

    Next, suppose \eqref{eqn:dot_gi} is $L$-RTI. 
    Then, \eqref{eqn:cond_g12Ltilde_autonomous} should be satisfied from \Cref{thm:rel_sti_L}. 
    Substituting \eqref{eqn:cor3_0} into \eqref{eqn:cond_g12Ltilde_autonomous},
    \begin{align*}
        h(\xi_2-\xi_1)g + (\xi_1 - \xi_2) hg = 0.
    \end{align*}
    Setting $g=e$, this is rearranged into $h(\xi_2-\xi_1) = (\xi_2-\xi_1)h$, or $\Ad_h (\xi_2-\xi_1)=\xi_2-\xi_1$.
    Since this should be satisfied for any $h,g\in\G$, it follows that $\xi_1\equiv\xi_2$. 

    Finally, substituting $\xi_1\equiv\xi_2$ into \eqref{eqn:dot_g12_L_tilde} yields \eqref{eqn:dot_g12_L_cor}.
    These shows (i).

    Next, (ii) can be shown similarly using \Cref{thm:rel_sti_R}.
\end{proof}

Therefore, the state trajectory independence of the relative state can be verified from the state trajectory independence of the estimation error, provided that the assumptions of \Cref{cor:ETI_RTI} are satisfied. 
For the given IMU example, the results of \Cref{ex:L-RTI,ex:R-RTI} can be reconstructed using \Cref{cor:ETI_RTI} as follows. 

\begin{example}{(RTI)}
    Consider the system model from \Cref{ex:UAV1}.
    It has been shown in \Cref{ex:ETI} that its estimation error is state trajectory independent (ETI).
    Furthermore, for any $ t $, we have $\tilde X_t(g) = (0,0,v)$ where $ g = (R,v,x) $.
    Thus, \Cref{cor:ETI_RTI} can be applied to this system.

    Let $g_1 = (R_1, v_1, x_1)$ and $g_2 = (R_2, v_2, x_2)$ such that the left-invariant relative state $g_{12}=g_1^{-1} g_2$ is obtained by
    \begin{align*}
        g_{12} & = (R_1^T R_2, R_1^T(v_2 - v_1), R_1^T(x_2 -x_1)),
    \end{align*}
    which is defined to be $g_{12} = (R_{12}, v_{12}, x_{12})$. 
    Since $ \xi_1(t) \equiv \xi_2(t) = (0, \grav e_3, 0) $, the left-invariant relative state $g_{12}$ is state trajectory independent ($L$-RTI), and from \eqref{eqn:dot_g12_L_cor} it is governed by
    \begin{align*}
        \dot g_{12} & =  - (\hat\Omega_1, a_1, 0) (R_{12}, v_{12}, x_{12}) + (0, 0, v_{12})\\
                    & \quad + (R_{12}, v_{12}, x_{12}) (\hat\Omega_2, a_2, 0),
    \end{align*}
    which is identical to \eqref{eqn:dot_g12_L_ex}.

    On the other hand, $ \zeta_1(t) \neq \zeta_2(t) $ since $\zeta_i(t) = (\hat \Omega_i(t) , a_i(t), 0)$. Thus the right-invariant relative state $g_{12}=g_1g_2^{-1}$ is not state trajectory independent (i.e., not $R$-RTI).

\end{example}

\subsection{Estimation Error for Left-Invariant Relative State ($L$-RETI)}

Once the state trajectory independence of the relative state is verified, we may construct an invariant filter for the relative dynamics.  
Since there are two types of relative states, we first assume that the left-invariant relative state $g_{12}=g_1^{-1} g_2$ of \eqref{eqn:dot_gi} is state trajectory independent ($L$-RTI), i.e., it is governed by \eqref{eqn:dot_g12_L_tilde} according to \Cref{thm:rel_sti_L}.  
Then, as discussed above, an estimator of the relative motion can be constructed on $\G$ without the need to estimate $g_1$ or $g_2$ individually.

Let $\bar g_{12}$ be another trajectory of \eqref{eqn:dot_g12_L_tilde} on $\G$, driven by the same input $u_{12}$ but potentially from a different initial condition; that is, $\bar g_{12}$ represents an estimated relative trajectory.  
Specifically,
\begin{align}
    \dot {\bar g}_{12} = (\xi_2 - \xi_1 - \zeta_2)\bar g_{12} + \tilde X(\bar g_{12}, u_2) + \bar g_{12}\zeta_2. \label{eqn:dot_bar_g12_L}
\end{align}

Define the left-invariant estimation error $f_{12} = g_{12}^{-1}\bar g_{12}$ and the right-invariant estimation error $h_{12} = \bar g_{12} g_{12}^{-1}$.  
If either $f_{12}$ or $h_{12}$ is state trajectory independent according to \Cref{def:sti}, then the relative motion can be estimated by an invariant Kalman filter.  
The resulting state trajectory independence of the relative estimation error is denoted by RETI.  
To avoid confusion, the definitions of all acronyms are summarized in \Cref{fig:summary}.

In short, we aim to verify the state trajectory independence of the estimation error for the left-invariant relative trajectory ($L$-RETI), given that the relative trajectories themselves are state trajectory independent ($L$-RTI), i.e., $L$-RTI $\overset{?}{\Rightarrow}$ $L$-RETI.  
Interestingly, this implication holds automatically without any additional requirements, as discussed below.  

\begin{theorem}{($L$-RETI)}\label{thm:rel_err_L_sti}
    If the left-invariant relative state $g_{12}=g_1^{-1}g_2$ of \eqref{eqn:dot_gi} is state trajectory independent ($L$-RTI) as presented in \Cref{def:rel_sti}, then the estimation error $f_{12}=g_{12}^{-1}\bar g_{12}$ or $h_{12} = \bar g_{12}g_{12}^{-1}$ of the resulting equation \eqref{eqn:dot_g12_L_tilde} for the relative state is also state trajectory independent ($L$-RETI) as defined in \Cref{def:sti}.  
    The corresponding error dynamics are given by
    \begin{align}
        \dot f_{12} & = \tilde X(f_{12}, u_2) + f_{12}\zeta_2 - \zeta_2 f_{12}, \label{eqn:dot_f12_L}\\
        \dot h_{12} & = \tilde X(h_{12}, u_2) + \xi_{12} h_{12} - h_{12} \xi_{12}, \label{eqn:dot_h12_L}
    \end{align}
    where $\xi_{12}(t) = \xi_2(t) - \xi_1(t) - \zeta_1(t) \in \g$.
\end{theorem}

\begin{proof}
    Since $g_{12}=g_1^{-1} g_2$ is $L$-RTI, according to \Cref{thm:rel_sti_L}, \eqref{eqn:cond_g12Ltilde_autonomous} holds.  
    Setting $g=e$ in \eqref{eqn:cond_g12Ltilde_autonomous} gives
    \begin{align*}
        \tilde X(h,u_2) = h\tilde X(e, u_2) + \tilde X(h, u_1) + h(\xi_2-\xi_1) + (\xi_1-\xi_2)h.
    \end{align*}
    Substituting $\tilde X(e, u_2)=0$ and multiplying both sides by $g$ on the right yields
    \begin{align*}
        \tilde X(h,u_2)g = \tilde X(h, u_1)g + h(\xi_2-\xi_1)g + (\xi_1-\xi_2)hg.
    \end{align*}
    Substituting this back into \eqref{eqn:cond_g12Ltilde_autonomous} leads to
    \begin{align}
        \tilde X(hg, u_2) = h\tilde X(g, u_2) + \tilde X(h, u_2)g.\label{eqn:thm4_0}
    \end{align}
    According to (iii) of \Cref{thm:groupaff}, \eqref{eqn:dot_g12_L_tilde} is ETI if and only if \eqref{eqn:group_aff_tilde} holds for $\tilde X(g_{12}, u_2)$, which is precisely \eqref{eqn:thm4_0}.  
    Therefore, $f_{12}$ and $h_{12}$ are ETI, corresponding to the $L$-RETI property.  
    Equations \eqref{eqn:dot_f12_L} and \eqref{eqn:dot_h12_L} then follow directly from \Cref{cor:dot_f_h}.
\end{proof}

In the above proof, the equivalent condition \eqref{eqn:group_aff_tilde} for state trajectory independence of the error in the decomposed form was instrumental.  
Because the relative-state dynamics are already expressed in a decomposed form as in \eqref{eqn:dot_g12_L_tilde}, ETI could be verified without explicitly considering the first and third terms of that equation.  

\Cref{thm:rel_err_L_sti} is particularly useful in practice: once the evolution of the relative state is shown to be self-contained, the relative state can be estimated by an invariant Kalman filter without any additional assumptions.  

\begin{example}{($L$-RETI)}
    The left-invariant relative state $g_{12}=g_1^{-1}g_2$ of the system in \Cref{ex:UAV1} is state trajectory independent ($L$-RTI), as shown in \Cref{ex:L-RTI}.  
    From \Cref{thm:rel_err_L_sti}, the estimation error of this relative state is also state trajectory independent.

    Using \eqref{eqn:dot_f12_L} with $f_{12} = (Q, u, y)$, the left-invariant estimation error evolves according to
    \begin{align*}
        \dot f_{12} & = (0, 0, u) + (Q, u, y)(\hat \Omega_2, a_2, 0)
                    - (\hat \Omega_2, a_2, 0)(Q, u, y) \\
                    & = (Q\hat\Omega_2 - \hat\Omega_2 Q,\; Qa_2 - \hat\Omega_2 u - a_2,\; u - \hat\Omega_2 y).
    \end{align*}

    Similarly, from \eqref{eqn:dot_h12_L}, the right-invariant estimation error $h_{12} = (T, p, q)$ satisfies
    \begin{align*}
        \dot h_{12} & = (0, 0, p) + (T, p, q)(\hat \Omega_1, a_1, 0)
                    - (\hat \Omega_1, a_1, 0)(T, p, q) \\
                    & = (T\hat\Omega_1 - \hat\Omega_1 T,\; T a_1 - \hat\Omega_1 p - a_1,\; p - \hat\Omega_1 q),
    \end{align*}
    since $\xi_{12} = \xi_2 - \xi_1 - \zeta_1 = -\zeta_1$ in this example.
\end{example}

\subsection{Estimation Error for Right-Invariant Relative State ($R$-RETI)}

We can derive similar results for the right-invariant relative state $g_{12}=g_1g_2^{-1}$.  
Suppose that it is state trajectory independent ($R$-RTI), i.e., it is governed by \eqref{eqn:dot_g12_R_tilde} according to \Cref{thm:rel_sti_R}.  
Let $\bar g_{12}$ be another trajectory of \eqref{eqn:dot_g12_R_tilde} on $\G$, driven by the same input $u_{12}$ but potentially from a different initial condition; that is, $\bar g_{12}$ represents an estimated relative trajectory.  
Specifically, we have
\begin{align}
    \dot{\bar g}_{12} = \xi_1 \bar g_{12} + \tilde X(\bar g_{12}, u_1) + \bar g_{12}(\zeta_1 - \zeta_2 - \xi_2). \label{eqn:dot_bar_g12_R}
\end{align}
Then, we can conclude that the corresponding estimation error is state trajectory independent.  

\begin{theorem}{($R$-RETI)}\label{thm:rel_err_R_sti}
    If the right-invariant relative state $g_{12}=g_1g_2^{-1}$ of \eqref{eqn:dot_gi} is state trajectory independent ($R$-RTI) as presented in \Cref{def:rel_sti}, then the estimation error $f_{12}=g_{12}^{-1}\bar g_{12}$ or $h_{12}=\bar g_{12}g_{12}^{-1}$ of the resulting equation \eqref{eqn:dot_g12_R_tilde} for the relative state is also state trajectory independent ($R$-RETI) as defined in \Cref{def:sti}.  
    The corresponding error dynamics are given by
    \begin{align}
        \dot f_{12} & = \tilde X(f_{12}, u_1) + f_{12}\zeta_{12} - \zeta_{12} f_{12}, \label{eqn:dot_f12_R}\\
        \dot h_{12} & = \tilde X(h_{12}, u_1) + \xi_1 h_{12} - h_{12}\xi_1, \label{eqn:dot_h12_R}
    \end{align}
    where $\zeta_{12}(t) = \zeta_1 - \zeta_2 - \xi_2 \in \g$.
\end{theorem}

\begin{proof}
    Since $g_{12}=g_1 g_2^{-1}$ is $R$-RTI, according to \Cref{thm:rel_sti_R}, \eqref{eqn:cond_g12Rtilde_autonomous} is satisfied.  
    Setting $h=e$ in \eqref{eqn:cond_g12Rtilde_autonomous} yields
    \begin{align*}
        \tilde X(g,u_1) = \tilde X(g,u_2) + \tilde X(e,u_1)g + (\zeta_1-\zeta_2)g + g(\zeta_2-\zeta_1).
    \end{align*}
    Substituting $\tilde X(e,u_1)=0$ and multiplying both sides by $h$ on the left gives
    \begin{align*}
        h\tilde X(g,u_1) = h\tilde X(g,u_2) + h(\zeta_1-\zeta_2)g + hg(\zeta_2-\zeta_1).
    \end{align*}
    Substituting this back into \eqref{eqn:cond_g12Rtilde_autonomous} leads to
    \begin{align}
        \tilde X(hg,u_1) = h\tilde X(g,u_1) + \tilde X(h,u_1)g. \label{eqn:thm5_0}
    \end{align}
    By (iii) of \Cref{thm:groupaff}, \eqref{eqn:dot_g12_R_tilde} is ETI if and only if \eqref{eqn:group_aff_tilde} holds for $\tilde X(g_{12},u_1)$, which corresponds exactly to \eqref{eqn:thm5_0}.  
    Therefore, $f_{12}$ and $h_{12}$ are ETI.  
    The dynamics in \eqref{eqn:dot_f12_R} and \eqref{eqn:dot_h12_R} then follow directly from \Cref{cor:dot_f_h}.
\end{proof}

\begin{figure}
    \footnotesize
    \tikzstyle{block} = [draw, fill=white, rectangle, minimum height=7.2em, minimum width=6em, text width=3.5cm]
    \begin{tikzpicture}
        \node[block, thick] (thm1) at (4.3,0) {\textbf{\Cref{thm:groupaff} $\Leftrightarrow$ \Cref{cor:dot_f_h} (ETI)}\\Estimation error $f=g^{-1}\bar g$ or $h=\bar g g^{-1}$ is state trajectory independent};
        \node[block] (thm2) at (0,-2.6) {\textbf{\Cref{thm:rel_sti_L} ($L$-RTI)}\\Left-invariant relative state $g_{12}=g_1^{-1}g_2$ is state trajectory independent};
        \node[block] (thm3) at (0,-4.7) {\textbf{\Cref{thm:rel_sti_R} ($R$-RTI)}\\Right-invariant relative state $g_{12}=g_1g_2^{-1}$ is state trajectory independent};
        \node[block, thick] (thm4) at (4.3,-2.6) {\textbf{\Cref{thm:rel_err_L_sti} ($L$-RETI)}\\Estimation error $f_{12}=g_{12}^{-1}\bar g_{12}$ or $h_{12}=\bar g_{12}g_{12}^{-1}$ of the left-invariant relative state $g_{12}=g_1^{-1}g_2$ is state trajectory independent};
        \node[block, thick] (thm5) at (thm3-|thm4) {\textbf{\Cref{thm:rel_err_R_sti} ($R$-RETI)}\\Estimation error $f_{12}=g_{12}^{-1}\bar g_{12}$ or $h_{12}=\bar g_{12}g_{12}^{-1}$ of the right-invariant relative state $g_{12}=g_1g_2^{-1}$ is state trajectory independent};
        \node[left=0.4cm of thm2.north] (tmp1) {};
        \node[above=0.4cm of thm1.west] (thm11) {};
        \draw[double, dotted, <-] (tmp1) -- (tmp1|-thm11) -- (thm11) node[pos=0.5,above] {\textbf{\Cref{cor:ETI_RTI}}};
        \node[right=0.4cm of thm2.north] (tmp2) {};
        \node[below=0.4cm of thm1.west] (thm12) {};
        \draw[double, ->] (tmp2) -- (tmp2|-thm12) -- (thm12) node[pos=0.5,above] {\textbf{\Cref{cor:RTI_ETI}}};
        \draw[double, ->] (thm2) -- (thm4) {};
        \draw[double, ->] (thm3) -- (thm5) {};
        \begin{scope}[on background layer]
            \node[inner sep=4pt, fill=gray!20, draw, fit=(thm2) (thm3)] (thm23) {};
        \end{scope}
        \begin{scope}[on background layer]
            \node[inner sep=4pt, fill=gray!20, draw, fit=(thm4) (thm5), densely dashdotted] (thm45) {};
        \end{scope}
        \node[rotate=90] (label1) at (-2.2,0) {Single trajectory};
        \node[rotate=90] at (label1|-thm23) {Relative trajectory};
        \node at (thm23.south) [below=1pt] {Relative state};
        \node at (thm45.south) [below=1pt] {Estimation error};
    \end{tikzpicture}
    \caption{Summary of the results for state trajectory independence.  
    The top row corresponds to a single trajectory, while the two shaded rows represent relative trajectories.  
    The left column depicts the relative states, and the right column corresponds to their estimation errors.  
    Thick lines indicate state trajectory independence of estimation errors.  
    A solid double arrow denotes a logical implication, whereas a dotted arrow indicates that additional conditions are required for the implication.}
    \label{fig:summary}
\end{figure}

All results on state trajectory independence are summarized in \Cref{fig:summary}.  
Our objective is to ensure the state trajectory independence of the estimation error for the relative state, i.e., $L$-RETI or $R$-RETI (the dash-dotted box in \Cref{fig:summary}), enabling the construction of an invariant filter for the relative states.  
The diagrams in \Cref{fig:summary} illustrate two possible pathways to achieve this.  
The state trajectory independence of the relative state can first be verified using \Cref{thm:rel_sti_L,thm:rel_sti_R}, and the resulting implications then automatically lead to the state trajectory independence of its estimation error.  
Alternatively, if the estimation error of a single trajectory is state trajectory independent, the additional conditions in \Cref{cor:ETI_RTI} allow us to follow the same reasoning.  
The latter approach is often more convenient, as it circumvents the explicit verification conditions required by \Cref{thm:rel_sti_L,thm:rel_sti_R}.

\section{Invariant Filter for Relative Dynamics}\label{sec:IKFRM}

In this section, we construct invariant filters for the relative dynamics, assuming that the state trajectory independence of the estimation error is satisfied as presented in \Cref{sec:STIRM}.  
To illustrate the development of the relative invariant filter clearly, we focus on a specific case in which the assumptions of \Cref{cor:ETI_RTI} hold for the state trajectory independence of the left-invariant relative state ($L$-RETI).  
However, we now incorporate process noise, which was excluded in \Cref{sec:STIRM}.  

\subsection{Relative Dynamics}

Suppose that the dynamics of two systems are governed by
\begin{align}
    \dot g_i = \xi_i(t) g_i + \tilde X(g_i) + g_i \zeta_i(t) + g_i \hat w_i, \label{eqn:dot_gi_w}
\end{align}
where $w_i(t)\sim \mathcal{N}(0,\Sigma_i)$ is an $n$-dimensional Gaussian white noise process with covariance $\Sigma_i\in\Re^{n\times n}$ for $i\in\{1,2\}$, which is mapped to $\g$ via the hat operator.  
The noise processes $w_1$ and $w_2$ are assumed to be independent.  

\begin{assumption}\label{assump:1}
    The terms $\tilde X$ and $\xi_i$ in \eqref{eqn:dot_gi_w} satisfy the following properties:
    \begin{enumerate}
        \renewcommand{\labelenumi}{(\roman{enumi})}
        \item $\tilde X(hg) = h\tilde X(g) + \tilde X(h) g$ for any $g,h\in\G$.
        \item $\xi_1(t)=\xi_2(t)$ for all $t$.
    \end{enumerate}
\end{assumption}

Under this formulation of the dynamics and the above assumptions, the results of \Cref{sec:STIRM} apply directly, ensuring the feasibility of constructing invariant filters for the relative dynamics.  
\begin{itemize}
    \item From \Cref{thm:groupaff}, property (i) implies that the error state is state trajectory independent (ETI).
    \item Combined with (ii), \Cref{cor:ETI_RTI} guarantees that the left-invariant relative state $g_{12}=g_1^{-1}g_2$ is state trajectory independent ($L$-RTI). 
    \item Finally, by \Cref{thm:rel_err_L_sti}, the corresponding estimation errors $f_{12}=g_{12}^{-1}\bar g_{12}$ and $h_{12}=\bar g_{12} g_{12}^{-1}$ are state trajectory independent ($L$-RETI). 
\end{itemize}

Since these results were originally derived in the noise-free setting, we now generalize the equations governing the relative state and its estimation error to include additive noise.  
Specifically, \eqref{eqn:dot_g12_L_cor} is augmented with noise terms as
\begin{align}
    \dot g = -(\zeta_1 + \hat w_1) g + \tilde X(g) + g (\zeta_2 + \hat w_2), \label{eqn:dot_g_L_cor}
\end{align}
where, for brevity, we omit the subscript $12$ (i.e., $g=g_{12}$).  
Throughout this section, we similarly drop the subscript $12$ from $g$, $\bar g$, $f$, and $h$.

The estimated relative trajectory $\bar g$ is propagated by \eqref{eqn:dot_g_L_cor} in the absence of noise:
\begin{align}
    \dot{\bar g} = -\zeta_1 \bar g + \tilde X(\bar g) + \bar g \zeta_2. \label{eqn:dot_g_bar_L_cor}
\end{align}
The resulting left-invariant estimation error $f=g^{-1}\bar g$ and right-invariant estimation error $h=\bar g g^{-1}$ evolve according to
\begin{align}
    \dot f & = \tilde X(f) + f\zeta_2 - (\zeta_2+\hat w_2)f + f(\Ad_{\bar g^{-1}}\hat w_1), \label{eqn:dot_f_L}\\
    \dot h & = \tilde X(h) + h(\zeta_1+\hat w_1) - \zeta_1 h - (\Ad_{\bar g}\hat w_2)h, \label{eqn:dot_h_L}
\end{align}
which generalize \eqref{eqn:dot_f12_L} and \eqref{eqn:dot_h12_L} by incorporating the noise terms.  
Note that the state trajectory independence of $f$ is violated by $w_1$ through the last term of \eqref{eqn:dot_f_L}, which depends on $\bar g$.  
Similarly, the state trajectory independence of $h$ is affected by $w_2$.  
Conversely, $w_2$ and $w_1$ do not influence the independence of $f$ and $h$, respectively.

\begin{example}\label{ex:rel_noise}
    The vehicle kinematics model in \Cref{ex:UAV1} is extended with noise as follows.  
    For each $i\in\{1,2\}$,
    \begin{align}
        \dot x_i & = v_i,\\
        \dot v_i & = R_i(a_i - w^a_i) + \grav e_3,\\
        \dot R_i & = R_i(\Omega_i - w^g_i)^\wedge,
    \end{align}
    where $w^a_i \sim \mathcal{N}(0, \Sigma_i^a)$ and $w^g_i\sim\mathcal{N}(0, \Sigma_i^g)$ are the additive accelerometer and gyroscope noise, respectively, with $\Sigma_i^a, \Sigma_i^g\in\Re^{3\times 3}$.  

    This system follows the form of \eqref{eqn:dot_gi_w}, with $\tilde X(g_i)=(0,0,v_i)$, $\xi_1=\xi_2=(0,\grav e_3,0)$, and $\zeta_i=(\Omega_i,a_i,0)$.  
    The noise terms are $w_i=(-w^g_i,-w^a_i,0)$ and $\Sigma_i=\mathrm{diag}(\Sigma_i^g,\Sigma_i^a,0_{3\times3})\in\Re^{9\times9}$.  
    The properties in \Cref{assump:1} are trivially satisfied, allowing the construction of an invariant filter for the left-invariant relative state $g_{12}=g_1^{-1}g_2=(R_1^T R_2, R_1^T(v_2-v_1), R_1^T(x_2-x_1))\triangleq(R,v,x)$.  
    The components of $g_{12}=(R,v,x)$ have the following interpretations:
    \begin{itemize}
        \item $R\in\SO3$: rotation matrix transforming a vector from the second body-fixed frame $\mathcal{B}_2$ to the first frame $\mathcal{B}_1$, representing the orientation of $\mathcal{B}_2$ relative to $\mathcal{B}_1$. 
        \item $v\in\Re^3$: velocity of the second body relative to the first, expressed in $\mathcal{B}_1$.
        \item $x\in\Re^3$: position of the second body relative to the first, expressed in $\mathcal{B}_1$.
    \end{itemize}

    From \eqref{eqn:dot_g_L_cor}, the relative dynamics are given by
    \begin{align}
        \dot x & = -(\Omega_1 - w_1^g)^\wedge x + v,\\
        \dot v & = -(\Omega_1 - w_1^g)^\wedge v - (a_1 - w_1^a) + R(a_2 - w_2^a), \\
        \dot R & = -(\Omega_1 - w_1^g)^\wedge R + R(\Omega_2 - w_2^g)^\wedge.
    \end{align}
    The estimated state $\bar g=(\bar R,\bar v,\bar x)$ is propagated by these equations with $w_1^g=w_1^a=w_2^g=w_2^a=0$.  

    The left-invariant error is $f=(R^T\bar R, R^T(\bar v-v), R^T(\bar x-x))\triangleq(Q,u,y)$.  
    Since $\tilde X(f)=(0,0,u)$ and $\zeta_i=(\hat\Omega_i,a_i,0)$, from \eqref{eqn:dot_f_L}, the left-invariant estimation error evolves as
    \begin{align}
        \dot y & = u - (\Omega_2 - w_2^g)^\wedge y - Q\bar R^T(w_1^g)^\wedge\bar x, \label{eqn:dot_y}\\
        \dot u & = Q a_2 - (\Omega_2 - w_2^g)^\wedge u - (a_2 - w_2^a) \nonumber\\
               & \quad - Q\bar R^T\big((w_1^g)^\wedge \bar v + w_1^a\big), \label{eqn:dot_u}\\
        \dot Q & = Q\hat\Omega_2 - (\Omega_2 - w_2^g)^\wedge Q - Q(\bar R^T w_1^g)^\wedge. \label{eqn:dot_Q}
    \end{align}

    Likewise, the right-invariant error $h=(\bar R R^T, -\bar R R^T v + \bar v, -\bar R R^T x + \bar x)\triangleq(T,p,q)$ evolves as
    \begin{align}
        \dot q & = p - \hat\Omega_1 q + (\bar R w_2^g)^\wedge(q - \bar x), \label{eqn:dot_q}\\
        \dot p & = T(a_1 - w_1^a) - \hat\Omega_1 p - a_1 \nonumber\\
               & \quad + (\bar R w_2^g)^\wedge(p - \bar v) + \bar R w_2^a, \label{eqn:dot_p}\\
        \dot T & = T(\Omega_1 - w_1^g)^\wedge - \hat\Omega_1 T + (\bar R w_2^g)^\wedge T. \label{eqn:dot_T}
    \end{align}
\end{example}

\subsection{Propagation with the Left-Invariant Error $f$}\label{sec:prop_L}

Next, we develop an invariant extended Kalman filter, which consists of two steps: propagation and correction.  
Note that the filter can be formulated either using \eqref{eqn:dot_f_L} with $f$, or using \eqref{eqn:dot_h_L} with $h$.  
We first consider the propagation step for the left-invariant estimation error $f$.

From the definition of the estimation error, the actual relative state is given by
\begin{align*}
    g = \bar g f^{-1}.
\end{align*}
The objective of the Bayesian filter is to construct a statistical model for the distribution of $g$.  
Specifically, the key idea of an extended Kalman filter is that the distribution of $g$ is highly concentrated around the current estimate $\bar g$, such that the error $f$ remains close to the identity.  
In this case, $\bar g$ is treated as the mean state,\footnote{The notion of a mean in Euclidean space does not trivially generalize to a Lie group and requires additional definitions. Here, the mean group element is interpreted approximately as the mode, i.e., the element with the highest probability density.}  
and the error is approximated as $f = \exp(\vartheta) = e + \vartheta + \mathcal{O}(\|\vartheta\|^2)$ for $\vartheta \in \g$.  

Since $\g$ is isomorphic to a Euclidean space through the hat (or vee) map, we assume that $\vartheta$ follows a Gaussian distribution, $\vartheta^\vee \sim \mathcal{N}(0, P)$, with covariance $P \in \Re^{n \times n}$ and $n = \dim \G$.  
Thus, the resulting stochastic model of $g$ based on the left-invariant error $f$ is
\begin{align}
    g = \bar g \exp^{-1}(\vartheta) \approx \bar g (e - \vartheta),
\end{align}
with $\vartheta^\vee \sim \mathcal{N}(0, P)$.  
This is denoted by $g \sim \mathcal{N}_L(\bar g, P)$.

Let the time domain be discretized by a sequence $\{t_0, t_1, \ldots\}$.  
The value of a variable at $t = t_k$ is denoted by a subscript $k$.  
Given $g_k \sim \mathcal{N}_L(\bar g_k, P_k)$, the goal of the propagation step is to determine $(\bar g_{k+1}, P_{k+1})$ such that $g_{k+1} \sim \mathcal{N}_L(\bar g_{k+1}, P_{k+1})$.  
Although there is no guarantee that $g_{k+1}$ follows the same distribution form, we enforce this assumption via the small-perturbation approximation.

The mean is propagated by \eqref{eqn:dot_g_bar_L_cor}, and the covariance is propagated by linearizing \eqref{eqn:dot_f_L} about $f = e$.  
Substituting $f = e + \vartheta + \mathcal{O}(\|\vartheta\|^2)$ into \eqref{eqn:dot_f_L} and expanding in $\vartheta$, while neglecting higher-order terms in $(\vartheta, w_1, w_2)$, we obtain
\begin{align}
    \dot \vartheta = D\tilde X(e) \cdot \vartheta + \ad_\vartheta \zeta_2 - \hat w_2 + \Ad_{\bar g^{-1}} \hat w_1, \label{eqn:dot_vartheta_0}
\end{align}
where $D\tilde X(e) \in \g^*$ denotes the derivative of $\tilde X(f)$ with respect to $f$ evaluated at $f = e$.  
Since the first two terms are linear in $\vartheta$, there exists a matrix representation $A : \g \rightarrow \Re^{n \times n}$ such that
\begin{align}
    (D\tilde X(e) \cdot \vartheta + \ad_\vartheta \zeta_2)^\vee = A(\zeta_2)\, \vartheta^\vee. \label{eqn:A}
\end{align}
Similarly, we define $G : \G \rightarrow \Re^{n \times n}$ such that
\begin{align}
    (\Ad_{\bar g^{-1}} \hat w_1)^\vee = G(\bar g^{-1}) w_1. \label{eqn:G}
\end{align}
Substituting these into \eqref{eqn:dot_vartheta_0} yields the linearized dynamics
\begin{align}
    \dot \vartheta^\vee = A(\zeta_2)\, \vartheta^\vee - w_2 + G(\bar g^{-1}) w_1.
\end{align}
Accordingly, the covariance propagates as
\begin{align}
    \dot P = A(\zeta_2)P + P A^T(\zeta_2) + \Sigma_2 + G(\bar g^{-1}) \Sigma_1 G^T(\bar g^{-1}). \label{eqn:dot_P}
\end{align}

In summary, the propagation step is completed by integrating \eqref{eqn:dot_g_bar_L_cor} and \eqref{eqn:dot_P} to obtain $\mathcal{N}_L(\bar g_{k+1}, P_{k+1})$.  
An important property of the invariant filter is that the error dynamics \eqref{eqn:dot_f_L} depend on $\bar g$ only through the last term, ensuring that the linearization is not further degraded by estimation errors.  
In \eqref{eqn:dot_P}, this is reflected in the fact that the matrix $A$ does not depend on $\bar g$, i.e., it provides an exact linearization for any estimated state $\bar g$.  

\begin{example}\label{ex:9}
    Consider the dynamics of the left-invariant estimation error given by \eqref{eqn:dot_y}, \eqref{eqn:dot_u}, and \eqref{eqn:dot_Q}.  
    Let $\vartheta = (\vartheta_Q, \vartheta_u, \vartheta_y)^\wedge \in \sett$ with $\vartheta_Q, \vartheta_u, \vartheta_y \in \Re^3$.  
    Since $\tilde X(f) = (0, 0, u)$ for $f = (Q, u, y)$ and $\zeta_2 = (\Omega_2, a_2, 0)$, we have
    \begin{align*}
        D\tilde X(e) \cdot \vartheta + \ad_{\vartheta}\zeta_2 
            &= (0, 0, \vartheta_u)
             - (\hat\Omega_2 \vartheta_Q,\, \hat a_2 \vartheta_Q + \hat\Omega_2 \vartheta_u,\, \hat\Omega_2 \vartheta_y),
    \end{align*}
    which can be rearranged as $A(\zeta_2)\, \vartheta^\vee$, where $A : \sett \rightarrow \Re^{9 \times 9}$ is given by
    \begin{align}
        A(\zeta_2) = 
        \begin{bmatrix}
            -\hat\Omega_2 & 0 & 0 \\
            -\hat a_2 & -\hat\Omega_2 & 0 \\
            0 & I_{3\times3} & -\hat\Omega_2
        \end{bmatrix}. \label{eqn:A_ex}
    \end{align}
    Similarly, since $\bar g = (\bar R, \bar v, \bar x)$ and $w_1 = (-w_1^g, -w_1^a, 0)$, we obtain
    \begin{align*}
        \Ad_{\bar g^{-1}} \hat w_1 
            = (-\bar R^T w_1^g,\; -\bar R^T((w_1^g)^\wedge \bar v + w_1^a),\; -\bar R^T (w_1^g)^\wedge \bar x),
    \end{align*}
    which can be expressed as $G(\bar g^{-1}) w_1$, where $G : \SEtt \rightarrow \Re^{9 \times 9}$ is
    \begin{align}
        G(\bar g^{-1}) = 
        \begin{bmatrix}
            \bar R^T & 0 & 0 \\
            -\bar R^T \hat{\bar v} & \bar R^T & 0 \\
            -\bar R^T \hat{\bar x} & 0 & \bar R^T
        \end{bmatrix}. \label{eqn:G_ex}
    \end{align}
    Using these matrices, the covariance can be propagated by \eqref{eqn:dot_P}.
\end{example}

\subsection{Correction with the Left-Invariant Error $f$}\label{sec:cor_L}

Suppose that a measurement $z_{k+1}\in\Re^m$ becomes available at $t_{k+1}$ in the following form:
\begin{align}
    z_{k+1} = Z(g_{k+1}) + w_z, \label{eqn:z}
\end{align}
where $Z:\G\rightarrow\Re^m$, and $w_z\sim\mathcal{N}(0,\Sigma_z)$ is an $m$-dimensional additive Gaussian measurement noise with covariance $\Sigma_z\in\Re^{m\times m}$.

The objective of the correction step is to update the propagated distribution $\mathcal{N}_L(\bar g_{k+1}, P_{k+1})$ to the posterior distribution $\mathcal{N}_L(\bar g_{k+1}^+, P_{k+1}^+)$ conditioned on the measurement $z_{k+1}$.  
As the correction step is performed at the fixed time $t_{k+1}$, we drop the subscript $k+1$ for brevity.  
Although a single measurement vector is considered here, the formulation naturally extends to multiple observations.

Using $g=\bar g f^{-1}$, \eqref{eqn:z} can be rewritten as
\begin{align*}
    z = Z(\bar g f^{-1}) + w_z,
\end{align*}
which expands to
\begin{align}
    z = Z(\bar g) - D Z(\bar g)\cdot(\bar g\vartheta) + w_z + \mathcal{O}(\|\vartheta\|^2). \label{eqn:z_0}
\end{align}
For a given $\bar g$, the first-order expansion $D Z(\bar g)\cdot\vartheta:\g\rightarrow\Re^m$ is a linear operator in $\vartheta$.  
Therefore, there exists $H_L:\G\rightarrow\Re^{m\times n}$ such that
\begin{align}
    -D Z(\bar g)\cdot(\bar g\vartheta) = H_L(\bar g)\vartheta^\vee. \label{eqn:H_L}
\end{align}
Substituting this into \eqref{eqn:z_0} and using the independence of $\vartheta$ and $w_z$, we obtain
\begin{align}
    z - Z(\bar g) \sim \mathcal{N}(0, S_L),
\end{align}
where $S_L\in\Re^{m\times m}$ is given by
\begin{align}
    S_L = H_L(\bar g) P H_L^T(\bar g) + \Sigma_z.
\end{align}
The joint distribution of $\vartheta$ and the measurement residual is
\begin{align*}
    \begin{bmatrix}\vartheta^\vee \\ z - Z(\bar g)\end{bmatrix}
    \sim
    \mathcal{N}\!\left(
        \begin{bmatrix}0\\0\end{bmatrix},
        \begin{bmatrix}
            P & P H_L^T\\
            H_L P & S_L
        \end{bmatrix}
    \right).
\end{align*}

Let $\vartheta^+$ denote the posterior distribution of $\vartheta$ conditioned on $z - Z(\bar g)$.  
By applying the standard conditional Gaussian formula~\cite{tong2012multivariate}, we obtain
\begin{align}
    (\vartheta^+)^\vee = (\vartheta\,|\,z - Z(\bar g))^\vee
    \sim \mathcal{N}\big(L_L(z - Z(\bar g)),\,P^+\big), \label{eqn:vartheta_post}
\end{align}
where $L_L, P^+\in\Re^{n\times n}$ are
\begin{align}
    L_L & = P H_L^T S_L^{-1}, \label{eqn:K_gain}\\
    P^+ & = (I_{n\times n} - L_L H_L)P. \label{eqn:P_post}
\end{align}
These correspond to the standard Kalman gain and covariance update equations.  
The posterior group element $f^+$ is then obtained from $\vartheta^+$ via the exponential map, and the posterior mean is computed as
\begin{align}
    \bar g^+ = \bar g \exp^{-1}\!\big((L_L(z - Z(\bar g)))^\wedge\big). \label{eqn:g_bar_post}
\end{align}
In summary, given the measurement $z$ and the prior estimate $\mathcal{N}_L(\bar g, P)$, the update step yielding $\mathcal{N}_L(\bar g^+, P^+)$ is completed by \eqref{eqn:g_bar_post} and \eqref{eqn:P_post}.

\begin{example}[Left-invariant measurement $z_L$]\label{ex:zL_left}
    Suppose that the measurement $z^L\in\Re^{3+3N_b}$ is given by
    \begin{align}
        z^L =
        \begin{bmatrix} z_x^L\\ z_{b_1}^L \\ z_{b_2}^L \\ \vdots \end{bmatrix}
        =
        \begin{bmatrix} x \\ R b_1 \\ R b_2 \\ \vdots \end{bmatrix}
        + w_z, \label{eqn:z_ex}
    \end{align}
    where each $b_i\in\Re^3$ with $\|b_i\|=1$ for $i\in\{1,\ldots,N_b\}$ represents a unit vector fixed in $\mathcal{B}_2$ and expressed in $\mathcal{B}_1$.  
    In the absence of noise, $z_x^L = x = R_1^T(x_2 - x_1)$ represents the position of $\mathcal{B}_2$ relative to $\mathcal{B}_1$, and $z_{b_i}^L = R b_i$ represents the known direction $b_i$ (in $\mathcal{B}_2$) observed in $\mathcal{B}_1$.  
    Hence, the relative configuration is measured by a set of sensors mounted on $\mathcal{B}_1$.  


    While we can follow the derivation in \Cref{sec:cor_L},
    the particular structure of \eqref{eqn:z_ex} allows an alternative correction formulation using left-invariant measurements~\cite{barrau2015nonlinear}, as described below.  
    Specifically, the measurements in \eqref{eqn:z_ex} can be expressed as a left group action on fixed, known vectors.  
    Define a pseudo-measurement $\tilde z^L\in\Re^{5(N_b+1)}$ as
    \begin{align}
        \tilde z^L = \mathcal{G}(g) B + W_z, \label{eqn:z_ex_0}
    \end{align}
    where $\mathcal{G}(g)\in\Re^{5(N_b+1)\times5(N_b+1)}$ and $B, W_z\in\Re^{5(N_b+1)\times1}$ are
    \begin{gather}
        \mathcal{G}(g) = I_{(N_b+1)\times(N_b+1)} \otimes g,\\
        B = \begin{bmatrix} B_x \\ B_1 \\ \vdots \end{bmatrix}, \quad
        W_z = \Pi^T w_z,
    \end{gather}
    with fixed matrices $B_x, B_i\in\Re^{5\times1}$ and $\Pi\in\Re^{3(N_b+1)\times5(N_b+1)}$ defined as
    \begin{gather*}
        B_x = e_5,\quad
        B_i = \begin{bmatrix} b_i \\ 0 \\ 0 \end{bmatrix},\\
        \Pi = I_{(N_b+1)\times(N_b+1)} \otimes \begin{bmatrix} I_{3\times3} & 0_{3\times2} \end{bmatrix}.
    \end{gather*}
    Essentially, $\tilde z^L$ is obtained by augmenting each block of $z^L$ with two rows of $(0,1)$ or $(0,0)$.
    The original measurement is recovered by the projection
    \begin{align}
        z^L = \Pi \tilde z^L.
    \end{align}

    While this manipulation is simple, the structure of $\tilde z^L$ is useful for linearization.  
    Multiplying \eqref{eqn:z_ex_0} by $\mathcal{G}^{-1}(\bar g)$ and noting that $\bar g^{-1}g = f^{-1}$, we obtain
    \begin{align}
        \mathcal{G}^{-1}(\bar g) \tilde z^L
        = \mathcal{G}(f^{-1}) B + \mathcal{G}(\bar g^{-1}) W_z. \label{eqn:GZL}
    \end{align}
    Expanding the first term gives
    \begin{align*}
        \mathcal{G}(f^{-1}) B
        &\approx B + \tilde H_L \vartheta^\vee,
    \end{align*}
    where $\tilde H_L\in\Re^{5(N_b+1)\times9}$ is
    \begin{align}
        \tilde H_L =
        \begin{bmatrix}
            0 & 0 & -I_{3\times3} \\
            0_{2\times3} & 0_{2\times3} & 0_{2\times3} \\
            \hat b_1 & 0 & 0 \\
            0_{2\times3} & 0_{2\times3} & 0_{2\times3} \\
            \vdots & \vdots & \vdots
        \end{bmatrix}.
    \end{align}
    Note that this expression is independent of the estimated state $\bar g$.  
    Substituting into \eqref{eqn:GZL} and projecting with $\Pi$, we obtain
    \begin{align*}
        \Pi\!\big(\mathcal{G}(\bar g^{-1}) \tilde z^L - B\big)
        &\approx H_L \vartheta^\vee + D(\bar R^{-1}) w_z,
    \end{align*}
    where $H_L = \Pi \tilde H_L \in \Re^{3(N_b+1)\times9}$ and $D(\bar R) = I_{(N_b+1)\times(N_b+1)} \otimes \bar R$.  
    This follows $\mathcal{N}(0, S_L)$, with
    \begin{align}
        S_L = H_L P H_L^T + D(\bar R^{-1}) \Sigma_z D(\bar R). 
    \end{align}
    The corresponding correction step is
    \begin{align}
        \bar g^+ = \bar g \exp^{-1}\!\big((L_L \Pi(\mathcal{G}(\bar g^{-1}) Z^L - B))^\wedge\big), \label{eqn:g_bar_post_ex}
    \end{align}
    with
    \begin{align}
        L_L & = P H_L^T S_L^{-1},\\
        P^+ & = (I_{n\times n} - L_L H_L)P.
    \end{align}
\end{example}

\begin{example}[Right-invariant measurement $z_R$]\label{ex:zR_left}
    Consider another type of measurement $z^R\in\Re^{3+3N_b}$ given by
    \begin{align}
        z^R =
        \begin{bmatrix} z_x^R \\ z_{b_1}^R \\ z_{b_2}^R \\ \vdots \end{bmatrix}
        =
        \begin{bmatrix} -R^T x \\ R^T b_1 \\ R^T b_2 \\ \vdots \end{bmatrix}
        + w_z. \label{eqn:zR_ex}
    \end{align}
    Here, $z_x^R = R_2^T(x_1 - x_2)$ represents the position of $\mathcal{B}_1$ relative to $\mathcal{B}_2$ resolved in $\mathcal{B}_2$, and $z_{b_i}^R = R_2^T R_1 b_i$ denotes the direction $b_i$ fixed in $\mathcal{B}_1$ as measured in $\mathcal{B}_2$.  
    Thus, the relative configuration is measured by a set of sensors attached to $\mathcal{B}_2$.

    Expanding $Z(g)=Z(\bar g f^{-1})$ gives
    \begin{align*}
        Z(g) &\approx
        \begin{bmatrix}
            -\bar R^T \bar x + \vartheta_y + (\bar R^T \bar x)^\wedge \vartheta_Q \\
            \bar R^T b_1 - (\bar R^T b_1)^\wedge \vartheta_Q \\
            \bar R^T b_2 - (\bar R^T b_2)^\wedge \vartheta_Q \\
            \vdots
        \end{bmatrix}
        = Z(\bar g) + H_L(\bar g)\vartheta^\vee,
    \end{align*}
    where $H_L:\SEtt\rightarrow\Re^{(3+3N_b)\times9}$ is
    \begin{align}
        H_L =
        \begin{bmatrix}
            (\bar R^T \bar x)^\wedge & 0 & I_{3\times3} \\
            -(\bar R^T b_1)^\wedge & 0 & 0 \\
            -(\bar R^T b_2)^\wedge & 0 & 0 \\
            \vdots & \vdots & \vdots
        \end{bmatrix}.
    \end{align}
    The correction step follows from \eqref{eqn:K_gain}, \eqref{eqn:P_post}, and \eqref{eqn:g_bar_post}.
\end{example}

\subsection{Propagation with the Right-Invariant Error $h$}\label{sec:prop_R}

The propagation and correction steps presented in \Cref{sec:prop_L,sec:cor_L} constitute an invariant extended Kalman filter for the relative dynamics using the left-invariant error $f$.  
Analogous results can be derived using the right-invariant error $h$, as outlined below.

From the definition of the estimation error $h = \bar g g^{-1}$, we have
\begin{align*}
    g = h^{-1} \bar g,
\end{align*}
where $h$ is expanded as $h = \exp(\varphi) = e + \varphi + \mathcal{O}(\|\varphi\|^2)$ for $\varphi\in\g$.  
It is assumed that $\varphi^\vee \sim \mathcal{N}(0, P)$ for $P\in\Re^{n\times n}$.  
This is denoted by $g \sim \mathcal{N}_R(\bar g, P)$.  
Since
\begin{align*}
    g = h^{-1}\bar g \approx (e - \varphi)\bar g = \bar g(e - \Ad_{\bar g^{-1}}\varphi),
\end{align*}
this is equivalent to
\begin{align*}
    g \sim \mathcal{N}_L\!\left(\bar g,\, \Ad_{\bar g^{-1}} P \Ad^T_{\bar g^{-1}}\right),
\end{align*}
where $\Ad_{\bar g^{-1}}$ is represented by an $n\times n$ matrix as in \eqref{eqn:Ad_g_inv}.

As before, the mean is propagated by \eqref{eqn:dot_g_bar_L_cor}, integrated over $[t_k, t_{k+1}]$ with $\bar g(t_k) = g_k$ to obtain $\bar g_{k+1} = \bar g(t_{k+1})$.  
Next, the covariance is propagated by linearizing \eqref{eqn:dot_h_L} about $h = e$.  
We substitute $h = e + \varphi + \mathcal{O}(\|\varphi\|^2)$ into \eqref{eqn:dot_h_L} and expand it in $\varphi$, while neglecting higher-order terms in $(\varphi, w_1, w_2)$, yielding
\begin{align}
    \dot\varphi = D\tilde X(e)\cdot\varphi + \ad_\varphi \zeta_1 + \hat w_1 - \Ad_{\bar g}\hat w_2. \label{eqn:dot_varphi_0}
\end{align}
Using $A:\g\rightarrow\Re^{n\times n}$ and $G:\G\rightarrow\Re^{n\times n}$ defined in \eqref{eqn:A} and \eqref{eqn:G}, respectively, this becomes
\begin{align}
    \dot\varphi^\vee = A(\zeta_1)\varphi^\vee + w_1 - G(\bar g) w_2.
\end{align}
Accordingly, the covariance evolves as
\begin{align}
    \dot P = A(\zeta_1)P + P A^T(\zeta_1) + \Sigma_1 + G(\bar g)\Sigma_2 G^T(\bar g). \label{eqn:dot_P_h}
\end{align}
In summary, the propagation step is completed by integrating \eqref{eqn:dot_g_bar_L_cor} and \eqref{eqn:dot_P_h} to obtain $\mathcal{N}_R(\bar g_{k+1}, P_{k+1})$.

\begin{example}\label{ex:11}
    Consider the dynamics of the right-invariant estimation error given by \eqref{eqn:dot_q}, \eqref{eqn:dot_p}, and \eqref{eqn:dot_T}.  
    Let $\varphi = (\varphi_T, \varphi_p, \varphi_q)^\wedge \in \sett$ for $\varphi_T, \varphi_p, \varphi_q \in \Re^3$.  
    We have $h = (T, p, q)$ and $\zeta_1 = (\Omega_1, a_1, 0)$.  
    From \eqref{eqn:A_ex},
    \begin{align}
        A(\zeta_1) =
        \begin{bmatrix}
            -\hat\Omega_1 & 0 & 0 \\
            -\hat a_1 & -\hat\Omega_1 & 0 \\
            0 & I_{3\times3} & -\hat\Omega_1
        \end{bmatrix}.
    \end{align}
    Next, with $\bar g = (\bar R, \bar v, \bar x)$ and $w_1 = (-w_1^g, -w_1^a, 0)$,  
    from \eqref{eqn:G_ex} we obtain
    \begin{align}
        G(\bar g) =
        \begin{bmatrix}
            \bar R & 0 & 0 \\
            \hat{\bar v}\bar R & \bar R & 0 \\
            \hat{\bar x}\bar R & 0 & \bar R
        \end{bmatrix}.
    \end{align}
    Using these expressions, the covariance can be propagated according to \eqref{eqn:dot_P_h}.
\end{example}

\subsection{Correction with the Right-Invariant Error $h$}\label{sec:cor_R}

Consider the measurement equation given by \eqref{eqn:z}.  
Using $g = h^{-1}\bar g$, \eqref{eqn:z} can be rewritten as
\begin{align*}
    z = Z(h^{-1}\bar g) + w_z,
\end{align*}
which expands as
\begin{align*}
    z = Z(\bar g) - D Z(\bar g)\cdot(\varphi \bar g) + w_z + \mathcal{O}(\|\varphi\|^2).
\end{align*}
Similar to \eqref{eqn:H_L}, there exists $H_R:\G\rightarrow\Re^{m\times n}$ such that
\begin{align}
    -D Z(\bar g)\cdot(\varphi \bar g) = H_R(\bar g)\varphi^\vee.
\end{align}
Following the same developments as in \Cref{sec:cor_L}, we obtain the posterior estimate
\begin{align}
    \bar g^+ = \exp^{-1}\!\big((L_R(z - Z(\bar g)))^\wedge\big)\,\bar g, \label{eqn:g_bar_post_R}
\end{align}
where
\begin{align}
    S_R & = H_R P H_R^T + \Sigma_z,\\
    L_R & = P H_R^T S_R^{-1},\\
    P^+ & = (I_{n\times n} - L_R H_R)P. \label{eqn:P_post_R}
\end{align}
In summary, given the measurement $z$ and the prior estimate $\mathcal{N}_R(\bar g, P)$,  
the correction step to obtain $\mathcal{N}_R(\bar g^+, P^+)$ is completed by \eqref{eqn:g_bar_post_R} and \eqref{eqn:P_post_R}.  
The propagation and correction steps presented in \Cref{sec:prop_R,sec:cor_R} together constitute an invariant extended Kalman filter for the relative dynamics using the right-invariant error $h$.

\begin{example}[Left-invariant measurement $z_L$]\label{ex:zL_right}
    Consider the left-invariant measurement $z_L$ defined in \eqref{eqn:z_ex}.  
    As in \Cref{ex:9}, let $\varphi = (\varphi_T, \varphi_p, \varphi_q)^\wedge$.  
    We have
    \begin{align*}
        g & = h^{-1}\bar g \approx (e - \varphi)\bar g\\
          & = (\bar R - \hat\varphi_T\bar R,\; \bar v - \hat\varphi_T\bar v - \varphi_p,\; \bar x - \hat\varphi_T\bar x - \varphi_q).
    \end{align*}
    The measurement function expands as
    \begin{align*}
        Z(g) = Z(h^{-1}\bar g)
        \approx
        \begin{bmatrix}
            \bar x - \hat\varphi_T\bar x - \varphi_q \\
            (I_{3\times3} - \hat\varphi_T)\bar R b_1 \\
            (I_{3\times3} - \hat\varphi_T)\bar R b_2 \\
            \vdots
        \end{bmatrix}
        = Z(\bar g) + H_R(\bar g)\varphi^\vee,
    \end{align*}
    where $H_R:\SEtt\rightarrow\Re^{(3+3N_b)\times9}$ is given by
    \begin{align}
        H_R(\bar g) =
        \begin{bmatrix}
            \hat{\bar x} & 0 & -I_{3\times3}\\
            (\bar R b_1)^\wedge & 0 & 0\\
            (\bar R b_2)^\wedge & 0 & 0\\
            \vdots & \vdots & \vdots
        \end{bmatrix}.
    \end{align}
    The correction step is then completed by \eqref{eqn:g_bar_post_R}--\eqref{eqn:P_post_R}.
\end{example}

\begin{example}[Right-invariant measurement $z_R$]\label{ex:zR_right}
    Consider the right-invariant measurement $z_R$ defined in \eqref{eqn:zR_ex}.  
%
    Let $\tilde z_R\in\Re^{5(N_b+1)}$ be a pseudo-measurement obtained by augmenting the first block of $z_R$ with $(0,1)$ and the remaining blocks with $(0,0)$.  
    Following similar developments as in \Cref{ex:zL_left}, we obtain
    \begin{align*}
        \Pi (\mathcal{G}(\bar g)\tilde z^R - B) \sim \mathcal{N}(0, S_R),
    \end{align*}
    where
    \begin{align}
        S_R = H_R P H_R^T + D(\bar R)\Sigma_z D(\bar R^{-1}),
    \end{align}
    with $H_R\in\Re^{3(N_b+1)\times9}$ given by
    \begin{align}
        H_R =
        \begin{bmatrix}
            0 & 0 & I_{3\times3} \\
            -\hat b_1 & 0 & 0 \\
            \vdots & \vdots & \vdots
        \end{bmatrix},
    \end{align}
    and $D(\bar R) = I_{(N_b+1)\times(N_b+1)} \otimes \bar R$.

    The corresponding correction step is
    \begin{align}
        \bar g^+ = \exp^{-1}\!\big((L_R \Pi(\mathcal{G}(\bar g)Z^R - B))^\wedge\big)\,\bar g, \label{eqn:g_bar_post_Rex}
    \end{align}
    where $L_R, P^+\in\Re^{n\times n}$ are
    \begin{align}
        L_R & = P H_R^T S_R^{-1},\\
        P^+ & = (I_{n\times n} - L_R H_R)P.
    \end{align}
\end{example}

As discussed in \Cref{ex:9}, let $\varphi = (\varphi_T, \varphi_p, \varphi_q)^\wedge$.  
The presented stochastic models, $\mathcal{N}_L$ and $\mathcal{
    N}_R$, are consistent with the concentrated Gaussian distributions described in~\cite{wang2006error}, which are based on the exponential coordinates of $\bar g^{-1}g = f^{-1}$ or $g\bar g^{-1} = h^{-1}$.  
While this work focuses on the extended Kalman filter formulation, the approach can be generalized to include second-order terms as in~\cite{wolfe2011bayesian}, or extended to Bayesian estimators that employ global probability distributions~\cite{lee2018bayesian,wang2021matrix}.  
Investigating the role of state trajectory independence in broader classes of estimators beyond the extended Kalman filter remains an interesting direction for future research.

\section{Numerical Example}\label{sec:Num}

In this section, we present numerical results for the running example of the vehicle kinematic model discussed in \Crefrange{ex:rel_noise}{ex:zR_right}.

\subsection{Simulation Parameters}

The relative estimator is simulated with the first vehicle following a circular trajectory and the second vehicle following a Lissajous curve.  
Specifically, the true position trajectories are given by
\begin{align*}
    x_1(t) &= \begin{bmatrix} 1.5 \cos (0.2 t) + 1 \\ 1.5\sin (0.2 t) + 1 \\ 0 \end{bmatrix}, \quad
    x_2(t) = \begin{bmatrix} 2 \sin(0.1 t) \\ 2\sin(0.2 t) \\ 2 \end{bmatrix},
\end{align*}
with the attitude chosen such that the first body-fixed axis is tangent to the position trajectory and the third body-fixed axis points downward.  
The simulation is conducted over a duration of $30$ seconds with a time step of $\Delta t = 0.01$ seconds.

For the process noise, we consider two types of IMUs, denoted by $\alpha$ and $\beta$, whose standard deviations for the gyro and accelerometer noise are
\begin{itemize}
    \item IMU--$\alpha$:
        \begin{align*}
            \sigma_\alpha^g &= \SI{6.108e-5}{rad/s/\sqrt{Hz}}, \\
            \sigma_\alpha^a &= \SI{1.373e-3}{m/s^2/\sqrt{Hz}}.
        \end{align*}
    \item IMU--$\beta$:
        \begin{align*}
            \sigma_\beta^g  &= \SI{1.2218e-3}{rad/s/\sqrt{Hz}}, \\
            \sigma_\beta^a  &= \SI{1.2355e-2}{m/s^2/\sqrt{Hz}}.
        \end{align*}
\end{itemize}
The corresponding covariance matrix for $\alpha$ is defined as
\[
    \Sigma_\alpha = \mathrm{diag}\big((\sigma_\alpha^g)^2 I_{3\times3},\, (\sigma_\alpha^a)^2 I_{3\times3},\, 0_{3\times3}\big) \in \Re^{9\times9},
\]
and $\Sigma_\beta$ is defined analogously.  
Thus, IMU--$\alpha$ is more accurate than IMU--$\beta$, and $\Sigma_\alpha \prec \Sigma_\beta$.

Measurements are available at \SI{2}{\hertz}.  
The position measurements are corrupted by additive Gaussian noise with a standard deviation of $\SI{0.5}{m}$.  
The fixed reference vectors for direction measurements are chosen as $b_1 = [1, 0, 0]^T$, $b_2 = [0, 1, 1]^T$, and $b_3 = [1, 1, 0]^T$, each normalized to unit length.  
The direction measurements are generated by rotating each true direction by a random rotation vector sampled from $\mathcal{N}(0_{3\times1}, \sigma_R^2 I_{3\times3})$ with $\sigma_R = \SI{5}{deg}$.  
The measurement covariance is then numerically computed from these sampled directions.

The initial estimate of the relative state is
\begin{align*}
    \bar g(0) = \big( & R_{12}(0)\exp(\tfrac{\pi}{2}\hat e_3),\, v_{12}(0) + [0.5,\, 0.8,\, -0.4]^T,\\
                      & x_{12}(0) + [0.0,\, 1.0,\, -1.0]^T \big),
\end{align*}
representing an initial error of \SI{90}{\degree} in attitude, \SI{1.02}{\meter/\second} in velocity, and \SI{1.41}{\meter} in position.  
The initial covariance for the filter with left-invariant error is
\[
    P_L(0) = \mathrm{diag}\big((\pi/6)^2 I_{3\times3},\, 2^2 I_{3\times3},\, 2^2 I_{3\times3}\big),
\]
while that for the filter with right-invariant error is obtained as
\[
    P_R(0) = [\Ad_{g(0)^{-1}}]\, P_L(0)\, [\Ad_{g(0)^{-1}}]^T.
\]

\subsection{Simulation Results}

We consider multiple cases depending on the type of estimator, measurement, and IMU.  
For the estimator type, we include the relative invariant filter with left-invariant error (\Cref{sec:prop_L,sec:cor_L}), denoted by LRKF, and the relative invariant filter with right-invariant error (\Cref{sec:prop_R,sec:cor_R}), denoted by RRKF.  
Additionally, a conventional quaternion-based extended Kalman filter (QEKF) is included as a baseline.

Two types of measurements are incorporated.  
The first is the left-invariant measurement $z_L$, presented in \Cref{ex:zL_left} for LRKF and in \Cref{ex:zL_right} for RRKF.  
The second is the right-invariant measurement $z_R$, presented in \Cref{ex:zR_left,ex:zR_right}, respectively for LRKF and RRKF.

Finally, based on the IMU type, we consider the following three cases:
\begin{itemize}
    \item Case I: $(\Sigma_1 = \Sigma_\alpha) \prec (\Sigma_2 = \Sigma_\beta)$,
    \item Case II: $(\Sigma_1 = \Sigma_\alpha) = (\Sigma_2 = \Sigma_\alpha)$,
    \item Case III: $(\Sigma_1 = \Sigma_\beta) \succ (\Sigma_2 = \Sigma_\alpha)$.
\end{itemize}

Hence, there are $3 \times 2 \times 3 = 18$ simulation cases in total.  
For each case, a Monte Carlo simulation with 1000 runs is performed, where all three filters are subjected to identical process and measurement noise for fair comparison.  
At each run, the estimation error is computed as
\begin{align*}
    \mathrm{Error}_k
    &= \|x_k - \bar x_k\| + \|v_k - \bar v_k\| \\
    &\quad + \frac{180}{\pi} \cos^{-1}\!\left(\frac{1}{2}(\mathrm{tr}(R_k^T \bar R_k) - 1)\right),
\end{align*}
where the subscript $k$ denotes the value of a variable at the $k$-th time step.  
This error is averaged over the simulation period, and its ensemble average over the Monte Carlo runs is denoted by $\mathrm{TotalError}$.

\begin{figure}[t]
    \centering
    \subfigure[Left-invariant measurement $z_L$]{%
        \includegraphics[width=0.8\columnwidth]{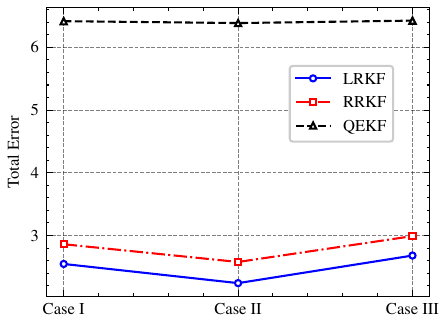}}
    \\
    \subfigure[Right-invariant measurement $z_R$]{%
        \includegraphics[width=0.8\columnwidth]{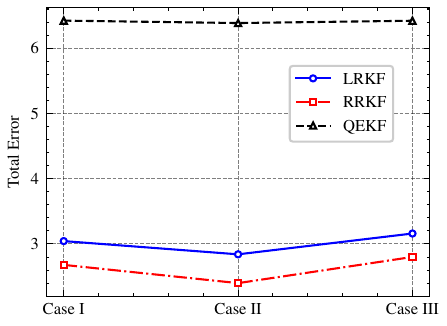}}
    \caption{Simulation results for different filter and measurement combinations.}
    \label{fig:sim}
\end{figure}

The error values for all 18 cases are illustrated in \Cref{fig:sim}, leading to the following observations:
\begin{itemize}
    \item The relative invariant filters (LRKF and RRKF) consistently outperform the quaternion-based extended Kalman filter (QEKF).
    \item Within the relative invariant filters, the correspondence between filter type and measurement type is crucial.  
    Specifically, for the left-invariant measurement $z_L$, LRKF achieves lower errors than RRKF, while for the right-invariant measurement $z_R$, RRKF outperforms LRKF.
    \item These trends hold consistently across all IMU types.  
    Differences in process noise between the two vehicles do not lead to statistically significant variations between LRKF and RRKF performance.  
    However, Case~II exhibits smaller errors for all filters, as the overall process noise is lower.
\end{itemize}

The superior performance observed when the filter and measurement share the same invariance type (i.e., LRKF with $z_L$ and RRKF with $z_R$) can be explained geometrically.  
In these matched cases, the estimation error evolves autonomously on the Lie group, independent of the true state trajectory---a property referred to as \emph{state trajectory independence}.
Consequently, the linearization of the error dynamics remains globally valid and does not degrade with the size of the estimation error.  
In contrast, when the filter and measurement invariances are mismatched, the estimation error dynamics acquire additional state-dependent terms that break this autonomy, leading to increased linearization error and reduced filter consistency.  
This alignment between the error structure and measurement invariance thus underlies the improved numerical performance observed in Fig.~\ref{fig:sim}.

\section{Conclusions}\label{sec:CL}

This paper developed a unified geometric framework for invariant filtering of relative dynamics on matrix Lie groups.  
The main theoretical contribution lies in deriving a new set of equivalent conditions for \emph{state trajectory independence} (STI) by uniquely decomposing the system vector field into components associated with left and right group actions.  
This decomposition provides a clear geometric characterization of when the estimation error dynamics become autonomous—independent of the true trajectory—and thus remain globally valid under large initial errors.

Building on this result, the derived STI conditions were systematically applied to construct invariant extended Kalman filters for relative dynamics.  
Both left- and right-invariant formulations were developed in a unified manner, yielding explicit propagation and correction steps that preserve the group structure while maintaining exact linearization of the error dynamics.  
The analysis further introduced the concept of \emph{relative trajectory independence} (RTI), which formalizes the autonomy of the relative estimation error between two systems.  
It was shown that RTI naturally implies STI for the relative error dynamics, ensuring that relative estimation remains consistent and independent of the individual trajectories.
The theoretical framework was validated through numerical simulations using a vehicle kinematic model.  

Beyond the specific case study, the results highlight the foundational role of trajectory independence in invariant filtering theory.  
The new STI and RTI formulations not only unify previous results on autonomous error dynamics but also provide a principled basis for designing symmetry-preserving estimators on nonlinear configuration spaces.  

Future work includes applying the proposed estimation scheme to the autonomous launch and recovery of unmanned aerial vehicles on ship decks in maritime environments, as well as extending it to formation estimation and control of multi-agent systems.

\bibliographystyle{IEEEtran}
\bibliography{ref_reduced}

\appendix

\subsection{Lie Group $\SEtt$}\label{sec:SEtt}

The group $\SEtt$ is a subgroup of the five-dimensional general linear group defined by
\begin{align*}
    \SEtt = \left\{ \aSEtt{R}{v}{x} \in \mathsf{GL(5)} \,\Bigg|\, R\in\SO3, v,x\in\Re^3 \right\}.
\end{align*}
For brevity, an element of $\SEtt$ is occasionally written as a tuple $g=(R,v,x)$. 
From the matrix multiplication, it is straightforward to show that the group operation is given by $g_1\circ g_2 = (R_1R_2, R_1v_2+v_1, R_1x_2 + x_1)$ for $g_i=(R_i,v_i,x_i)\in\SEtt$ with $i\in\{1,2\}$.
The inverse of $g=(R,v,x)$ is $g^{-1} = (R^T, -R^Tv, -R^Tx)$.

Next, its Lie algebra, $\sett$ is
\begin{align*}
    \sett = \left\{ \asett{\hat\Omega}{A}{V} \in \mathfrak{gl}(5) \,\Bigg|\, \Omega,A,V\in\Re^3 \right\}.
\end{align*}
This definition provides an isomorphism between $\sett$ and $\Re^9$, and accordingly an element of a Lie algebra may be denoted by a tuple $\eta=(\Omega, A, V)$. 
The Lie bracket or the ad operator is given by
\begin{align*}
    \ad_{\eta_1} \eta_2 = [\eta_1, \eta_2] = (\hat\Omega_1\Omega_2, \hat\Omega_1 A_2-\hat\Omega_2A_1, \hat\Omega_1V_2-\hat\Omega_2 V_1),
\end{align*}
for $\eta_i=(\Omega_i, A_i, V_i)$ with $i\in\{1,2\}$. 
Or in a matrix form, 
\begin{align*}
    \ad_{\eta_1} \eta_2 = \begin{bmatrix} \hat\Omega_1 & 0  & 0 \\ \hat A_1 & \hat\Omega_1  & 0 \\ \hat V_1 & 0 & \hat\Omega_1 \end{bmatrix}
    \begin{bmatrix} \Omega_2 \\ A_2 \\ V_2 \end{bmatrix}.
\end{align*}

The kinematics equation can be written as $\dot g = g\eta = \zeta g$, or equivalently
\begin{align*}
    (\dot R, \dot v, \dot x) = (R\hat\Omega, RA, RV) = (\hat\omega R, \hat \omega v + a, \hat\omega x + u)
\end{align*}
for $g=(R,v,x)\in\SEtt$ and $\eta=(\Omega,A,V), \zeta = (\omega, a, u) \in\sett$. 


The adjoint map $\Ad_g:\sett\rightarrow\sett$ is 
\begin{align}
    \Ad_g \eta & = g\eta g^{-1} 
                 = (R\Omega,  -\widehat{R\Omega} v + RA, -\widehat{R\Omega} x + RV) \nonumber \\
               & = \begin{bmatrix} 
                       R & 0 & 0 \\
                       \hat v R & R & 0 \\
                       \hat x R & 0 & R
                   \end{bmatrix} 
                   \begin{bmatrix} \Omega \\ A \\ V \end{bmatrix},\label{eqn:Ad_g} \\
    \Ad_{g^{-1}} \zeta & = g^{-1}\zeta g 
                         = (R^T\omega,  R^T(\hat\omega v + a), R^T(\hat\omega x + u)) \nonumber \\
                       & = \begin{bmatrix} 
                           R^T & 0 & 0 \\
                           -R^T\hat v & R^T & 0 \\
                           -R^T\hat x & 0 & R^T 
                       \end{bmatrix}
                       \begin{bmatrix} \omega \\ a \\ u \end{bmatrix}.\label{eqn:Ad_g_inv}
\end{align}

\begin{IEEEbiography}[{\includegraphics[width=1in,height =1.25in,clip,keepaspectratio]{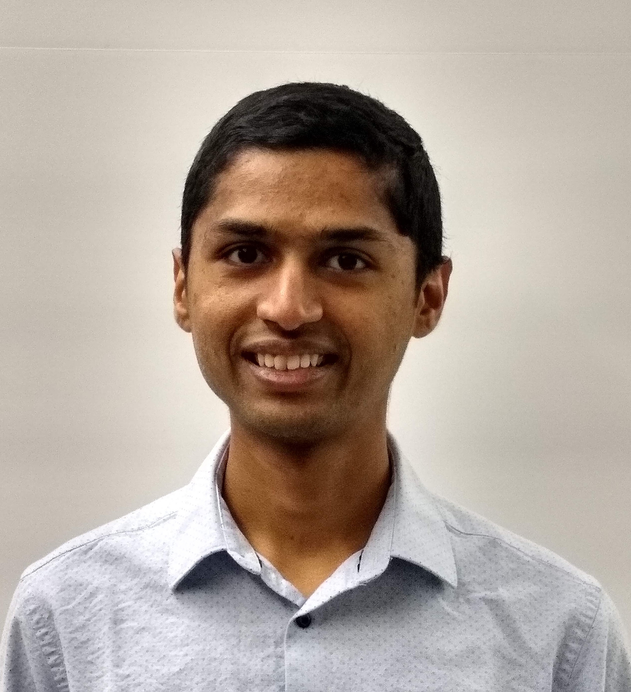}}]
    {Tejaswi K. C.} received the B.Tech. and M.Tech. degree in Aerospace Engineering from IIT Bombay, India, in 2019, and the Ph.D. in Mechanical and Aerospace Engineering at The George Washington University in 2024.
    He is currently a postdoctoral associate in Washington, D.C., U.S.A.
    His recent research consists of theoretical and computational methods in dynamics, optimization, data-driven methods and bio-inspired systems. 
\end{IEEEbiography}

\begin{IEEEbiography}[{\includegraphics[width=1in,height =1.25in,clip,keepaspectratio]{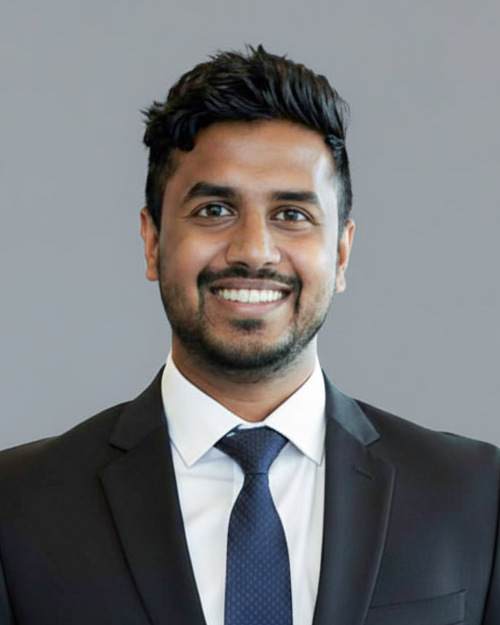}}]
	{Maneesha Wickramasuriya}
	is currently a Ph.D. candidate in mechanical and aerospace engineering at the George Washington University, Washington, D.C., USA. He received the B.Sc. degree in mechanical engineering from the University of Peradeniya, Sri Lanka, in 2017. His research interests include geometric control of unmanned aerial vehicles, deep computer vision–based navigation, and multi-sensor fusion for autonomous systems.
\end{IEEEbiography}. 

\begin{IEEEbiography}[{\includegraphics[width=1in,height =1.25in,clip,keepaspectratio]{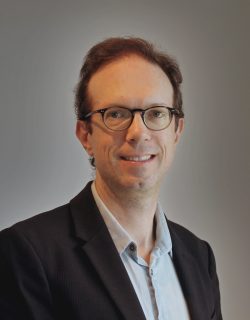}}]
    {Silv\`ere Bonnabel} (member, IEEE) received the engineering and Ph.D. degree from Mines Paris PSL, in 2004 and 2007. He is currently Professor at Mines Paris PSL. In 2017, he was Invited Fellow at the University of Cambridge. He serves as an Associate Editor of {IEEE Control Systems Magazine}. He is a recipient of various awards, including the 2015 IEEE SEE Glavieux Prize, the 2021  European Control Award, the 2022 Espoir Award from IMT-Academie des Sciences.
\end{IEEEbiography}

\vspace*{-2\baselineskip}

\begin{IEEEbiography}[{\includegraphics[width=1in,height =1.25in,clip,keepaspectratio]{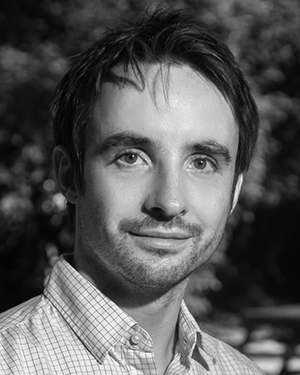}}]
    {Axel Barrau} received the engineering degree in applied mathematics from Ecole Polytechnique, Palaiseau, France, in 2013, and the Ph.D. degree in computer science from Ecole des Mines Paris  in 2015. He is currently CTO at OFFROAD.  Dr. Barrau was the recipient of the Automatica Paper Prize in 2020 and the 2024 George N. Saridis Best Transactions Paper Award for Outstanding Research from IEEE Transactions on Intelligent Vehicles. 
\end{IEEEbiography}

\begin{IEEEbiography}[{\includegraphics[width=1in,height =1.25in,clip,keepaspectratio]{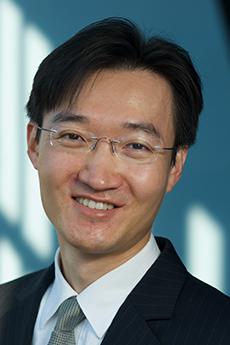}}]
    {Taeyoung Lee}
    		(Senior Member, IEEE) received the Ph.D. degree in aerospace engineering and the M.S. degree in mathematics from the University of Michigan, Ann Arbor, MI, USA, in 2008.
		He is a professor of the Department of Mechanical and Aerospace Engineering at the George Washington University, Washington, D.C., USA.
		His research interests include geometric mechanics and control with applications to complex aerospace systems.
\end{IEEEbiography}

\end{document}